\documentclass[twocolumn]{article}
\usepackage[letterpaper,margin=0.6in,columnsep=15pt]{geometry}
\usepackage{mathtools}
\usepackage{amsthm}
\usepackage{amssymb}
\usepackage{amsmath}
\usepackage{enumitem}
\usepackage{tabularx}
\usepackage{booktabs}
\usepackage{authblk}

\newcommand{\bits}{\{0,1\}}
\newcommand{\Setup}{\mathsf{Setup}}
\newcommand{\Init}{\mathsf{Init}}
\newcommand{\Authorize}{\mathsf{Authorize}}
\newcommand{\Judge}{\mathsf{Judge}}
\newcommand{\Can}{\mathsf{Can}}
\newcommand{\Fin}{\mathsf{Fin}}
\newcommand{\Acc}{\mathsf{Acc}}
\newcommand{\Emit}{\mathsf{Emit}}
\newcommand{\Requested}{\mathsf{Requested}}
\newcommand{\LAEUF}{\mathsf{LA\text{-}EUF}}
\newcommand{\PQEUF}{\mathsf{PQ\text{-}EUF\text{-}CMA}}
\newcommand{\LQROM}{\mathsf{LQROM}}
\newcommand{\Adv}{\mathsf{Adv}}
\newcommand{\ctx}{\mathsf{ctx}}
\newcommand{\CR}{\mathsf{CCR}}
\newcommand{\qPRFQRO}{\mathsf{qPRF\text{-}QRO}}
\newcommand{\Bad}{\mathsf{Bad}}
\newcommand{\Fresh}{\mathsf{Fresh}}
\newcommand{\Hit}{\mathsf{Hit}}
\newcommand{\Qadv}{Q_{\mathsf{adv}}}
\newcommand{\Qmed}{Q_{\mathsf{med}}}
\newcommand{\Qhon}{Q_{\mathsf{hon}}}
\newcommand{\Qeff}{Q_{\mathsf{eff}}}
\newcommand{\Qhead}{Q_{\mathsf{head}}}
\newtheorem{theorem}{Theorem}
\newtheorem{corollary}{Corollary}
\newtheorem{lemma}{Lemma}
\newtheorem{definition}{Definition}
\newtheorem{proposition}{Proposition}

\title{Finality Before Disclosure for Ledger Authenticators in the Quantum Random Oracle Model}

\author[1]{Maja Lie}
\author[1, 2]{Benjamin Marsh}
\affil[1]{Sei Labs}
\affil[2]{University of Portsmouth}
\date{August 2026}

\begin{document}

\maketitle

\begin{abstract}
Public ledgers increasingly authorize state transitions using prior transactions,
finalized state, timing, and ordering rather than only a public key, message, and
portable signature.  Standard unforgeability games omit the public pending pool,
adversarial scheduling, finality, and verification that depends on the
transcript.  We
introduce \emph{ledger authenticators} and $\LAEUF$, an unforgeability experiment
for reactive authorization protocols whose public judgment algorithm reads a
finalized transcript.  The model separates authentication safety from ledger
liveness and captures canonical transition freshness, adaptive corruption,
exposure before inclusion, censorship, and adversarial ordering.  We identify
two conditional resource boundaries.  An authenticator satisfying our single
event conditions yields a contextual one-time signature.  Within our
rebindable reveal class, safety requires computational post-disclosure
non-admissibility.  When precursor admission uses only public computation and
ledger scheduling, this condition is enforced by closing the evidence eligible
to use a disclosed credential.  If newly constructed evidence remains
admissible after disclosure, censoring the honest reveal gives a forgery.  We
then define a joint ledger and quantum random oracle execution model in which
quantum state persists across classical finalization cuts and oracle
evaluations made through the ledger are charged.  For a closed finalized target set of size at
most $K$, we prove the bound
$3\beta_{\mathsf{cut}}^2+3c_{\mathsf{co}}KQ^2/2^\lambda
+6\ell/2^\lambda$, where
$\beta_{\mathsf{cut}}$ accounts for fresh openings already present at the cut.
A commit, close, reveal authenticator instantiates the framework and obtains a
multi-user lifetime QROM bound.  For public admission rebindable
reveal authentication, eligibility closure in finalized state before
disclosure is the ledger supplied cryptographic resource.
\end{abstract}

\section{Introduction}\label{sec:introduction}

Post-quantum ledger authentication is usually framed as signature
substitution.  One replaces ECDSA or EdDSA with ML-DSA, SLH-DSA, or a stateful
hash based signature~\cite{nistfips204,nistfips205,nist800208}.  This preserves
a portable interface in which a verifier needs only a public key, a message,
and a signature. Such an interface ignores the environment in which ledger 
authentication occurs: the verifier need not know how the public key
was established, what other messages have been authorized, or whether multiple
parties have observed the same history. Ledger protocols, however, naturally
expose precisely this additional information and can trade portability for
prior transactions, finalized state, timing, and ordering. Guy Fawkes authentication
and Fawkescoin are early examples which exploit this trade~\cite{anderson1998guy,bonneau2014fawkescoin}.
Recent hash based proposals revisit this idea in metatransaction and
post-quantum settings~\cite{hughes2023hbauth,finlow2026cost}. These works suggest that a
ledger need not merely transport signatures; it can itself contribute to the
security argument by constraining when and how authentication evidence becomes
valid. This raises the paper's central question. Which part of authentication hardness
comes from a private credential, and which part comes from finalized history?
Registration and nonce checks use the ledger as a public key directory and
replay guard, but add no cryptographic hardness to a local authenticator. For a
rebindable reveal protocol, safety requires computational post-disclosure
non-admissibility. After learning the credential, no efficient scheduler may
make newly constructed evidence eligible for a different action. When
admission otherwise uses only public computation and ledger scheduling,
finalized state can enforce this condition by fixing the complete eligible set
before disclosure. An independent computational admission gate may instead
make post-disclosure admission infeasible under its own assumption. The usual signature game cannot express this distinction since it models
authentication as an isolated cryptographic primitive. It has no public
pending pool in which an honest credential becomes visible before inclusion,
no adversarial scheduler that may suppress that event and insert another, and
no distinction between inclusion and finality.  Message freshness is also too
weak when the same application call with a different successor credential
induces a different state transition.  In the quantum setting, the adversary
may retain quantum state across many authorizations and across the finalization
cut, while submitted events can cause oracle evaluations by the verifier.
Query budgets for each transaction or quantum search accounting begun only after
disclosure therefore miss part of the attack surface.

We make the verifier that reads the transcript the primary object.  A
\emph{ledger authenticator} has a reactive authorization algorithm and a
historical judgment predicate on finalized transcripts.  Its unforgeability
game exposes every honest emission before inclusion and gives scheduling and
censorship to the adversary.  Safety assumes no honest inclusion bound.
Completeness is stated separately under an explicit ledger liveness condition.
We show that mechanisms satisfying our single event conditions yield one-time
signatures. Within the rebindable reveal class, a mechanism whose admission is
otherwise public can use finalized history by closing admissible evidence
before its credential leaks.  In the QROM, that close is a classical proof cut,
not a reset.  Quantum state and preprocessing that depends on oracle values
persist across it, and evaluations by the verifier remain part of the lifetime
query budget.

\paragraph{Contributions.}
This paper makes three contributions.

\begin{enumerate}[leftmargin=1.5em]
  \item We define ledger authenticators and a multi-user $\LAEUF$
  experiment with reactive honest processes, exposure before inclusion,
  adversarial scheduling, adaptive corruption, finalized historical receipts,
  and canonical transition freshness.  Ordinary post-quantum signatures embed
  into the model.  Conversely, authentication satisfying our single event
  conditions in a fixed context yields a contextual one-time signature by a
  straight line reduction that preserves the stated resources.

  \item We isolate the closure condition for reveal protocols. Replay protection prevents duplicate acceptance of the same evidence, whereas our closure theorem prevents newly created evidence from becoming eligible after disclosure. If a disclosed
  credential can be rebound to another action, admitting newly constructed
  evidence after disclosure yields a censor and rebind forgery.  Thus safety
  requires negligible post-disclosure admission probability. When admission
  otherwise uses only public computation and scheduling, finalized eligibility
  closure supplies this property.

  \item We specify the ledger QROM and prove a finalized target search rule for
  adaptive classical cuts, persistent quantum state, public transcripts
  correlated with the oracle, and evaluations made through the ledger.  A commit, close, reveal
  authenticator applies the rule and yields a multi-user lifetime
  security bound together with parameter constraints for deployment.
\end{enumerate}

\section{Related work and positioning}\label{sec:related}

We position our contribution relative to four complementary research threads: portable signature security, temporal ledger authentication, formal ledger execution models, and quantum random oracle analysis.

\paragraph{Portable signatures.}
Classical EUF-CMA treats verification as a local predicate of a public key,
message, and signature~\cite{gmr1988signatures}.  Boneh and Zhandry study
signatures against quantum adversaries~\cite{bonehzhandry2013signatures}.
XMSS/WOTS+ and LMS/LM-OTS are standardized stateful hash based
signatures~\cite{hulsing2018xmss,mcgrew2019lms,nist800208}, while FIPS~205
standardizes stateless SLH-DSA~\cite{nistfips205}.  These objects remain
portable.  A ledger nonce can reject a second accepted signature at a reused
one-time index, but cannot prevent a wallet from exposing two pending
signatures at that index. Within the conditions of
Definition~\ref{def:local}, our single event extraction result marks the
boundary.
Local ledger transport and replay guards do not make public authentication
free.

\paragraph{Temporal authentication and ledger reveal protocols.}
Anderson et al.~\cite{anderson1998guy} authenticate a stream by placing a hash
commitment on a trusted timeline before revealing its preimage.  Fawkescoin
uses a blockchain for the same temporal role and already identifies the
censorship race after disclosure~\cite{bonneau2014fawkescoin}. These works establish the basic idea
that authentication can derive security not only from a secret credential but
also from an immutable ordering of events. Subsequent work applies similar timing mechanisms in different settings. Stewart et
al.~\cite{stewart2018slow} and Sattath and Wyborski~\cite{sattath2023lifted}
use related timing mechanisms for one-time migration or recovery after a
quantum break.  HBAuth studies inexpensive hash-based authorization for
batched metatransactions~\cite{hughes2023hbauth}.  Finlow-Bates et
al.~\cite{finlow2026cost} propose a compact two transaction hash construction
as an alternative to large post-quantum signatures and leave formal security
analysis to future work. Akshat~\cite{akshat2026commitreveal} studies EVM
commit-reveal authorization, including binding completeness, a quantum
front-running game, implementation, and gas costs. Its analysis concerns a
probabilistic post-finality race. We instead study transcript unforgeability
under indefinite censorship and persistent state lifetime QROM bounds. We do
not claim the first temporal remedy. Rather, we formalize
an unforgeability game for the full execution, a necessary closure
theorem for rebindable reveals, and a lifetime QROM analysis.  The resulting
condition is stronger than confirmation of one commitment because all evidence
that may use the disclosed credential must already be fixed by finality.

\paragraph{Ledger execution models.}
Universal composability, blockchain backbones, and composable ledger
abstractions model distributed execution and consensus~\cite{canetti2001uc,
garay2015backbone,badertscher2017ledger}.  LedgerLocks gives a framework for
protocols based on locked adaptor signature transactions~\cite{tairi2023ledgerlocks}.
Kaptchuk, Green, and Miers use a ledger to authenticate publication to
otherwise stateless computation~\cite{kaptchuk2019ledgers}.  $\LAEUF$ is
narrower than these frameworks.  It assumes a ledger interface and gives a
game for recurring account authorization whose public verifier
reads finalized history.

\paragraph{Quantum random oracles.}
Boneh et al.~\cite{boneh2011qrom} formulate random oracle security against
quantum queries, and Bennett et al.~\cite{bennett1997strengths} give the
quadratic search law behind our concrete bounds.  Compressed oracle and query
recording techniques make oracle knowledge usable in reductions without
recording every superposition query classically~\cite{zhandry2019record,
chung2021compressed}.  Our finalized target lemma is a proof
rule derived from this machinery.  It identifies the classical cut, the
immutable target set, the bad amplitude at the cut, and every evaluation made
through the ledger that must be charged.

\section{Ledger authenticators}\label{sec:model}

\subsection{Finalized ledger execution}

Time is divided into slots.  The ledger maintains a full transcript $\mathcal T$
and exposes a finalized prefix
$\mathcal T_t^{\mathsf{fin}}=\Fin_t(\mathcal T)$ at slot $t$.  A finalized
entry contains the event bytes, inclusion and finalization slots, the relevant
states before and after execution, and a receipt.  Except on an event
$\Bad_{\mathsf{fin}}$ of probability at most
$\varepsilon_{\mathsf{fin}}$, finalized prefixes are monotone such 
\(
 t\le t'\) implies \(
 \mathcal T_t^{\mathsf{fin}}
 \preceq \mathcal T_{t'}^{\mathsf{fin}}.
\) This abstraction covers deterministic and probabilistic finality by placing
the desired failure probability in $\varepsilon_{\mathsf{fin}}$.  An event has
three relevant moments.  It is first \emph{emitted} into the pending pool and
immediately delivered to the adversary.  It may later be \emph{included} in a
candidate history.  It may eventually enter the \emph{finalized} prefix.  The
adversary chooses which pending events are included and their order, subject to
the public state transition relation.  It may censor any honest event without
a time bound.  Consequently, safety cannot be proved by assuming that the
honest pending event wins a race.  For completeness only, we sometimes assume
$(\Delta_{\mathsf{inc}},\Delta_{\mathsf{fin}})$ liveness, under which an
admissible honest event emitted at $t$ is included by
$t+\Delta_{\mathsf{inc}}$, and an included honest event finalizes within
another $\Delta_{\mathsf{fin}}$ slots.  No such bound is available to an
unforgeability reduction, since an adversary may delay honest events arbitrarily.

All definitions and experiments are relative to a fixed public classical
polynomial-time ledger transition machine
\[
 \mathcal L=(\mathsf{Valid},\mathsf{Apply},\mathsf{Fin}).
\]
It fixes event parsing, state guards, deterministic application on a legal
candidate history, receipt creation, and finality. The adversary controls
submission, inclusion, ordering, and calls to finalization, but cannot change
these algorithms. The QROM environment $\mathcal E$ introduced later is the
interactive wrapper that runs $\mathcal L$, the honest algorithms, and the
common oracle. We write
$\Adv_{\mathsf{LA},\mathcal L}^{\LAEUF}$ when the transition machine matters
and suppress $\mathcal L$ when it is fixed by context.

\subsection{Canonical actions and historical receipts}

Let $a$ be an account, $e$ a key epoch, $i$ an authorization nonce, $b$ an
application request, and $\rho$ all authorization state installed by the
transition. Let
\[
 q=(\mathsf{chainId},\mathsf{forkId},a,e,i,b,\rho).
\]
A deterministic canonicalizer with length delimiters produces
\[
 \tau=\Can(q).
\]
The component $\rho$ includes a successor public key or hash chain head when
one is installed.  It also includes any expiry, fee, call root, or delegation
field that changes what the ledger will do.  Otherwise an adversary could keep
$b$ fixed, replace the successor credential, and hijack the account without
creating a fresh action in the security game. We require \emph{effect
injectivity}: if $\Can(q)=\Can(q')$, then, for every common reachable prestate
$S$, the two descriptions have the same guards and, when accepted, induce the
same transition under $\mathcal L$. Different prestates may produce different
poststates. An exact canonical encoding has
$\varepsilon_{\mathsf{can}}=0$.  An implementation that hashes an intent first
must add the appropriate collision or ambiguity failure to
$\varepsilon_{\mathsf{can}}$.

An accepted event creates a provisional receipt at its inclusion slot
$t_{\mathsf{inc}}$ and canonical position $p_{\mathsf{inc}}$. If the entry
later finalizes at slot $t_{\mathsf{fin}}$ and canonical position
$p_{\mathsf{fin}}$, the finalized transcript contains the historical receipt
\[
 \Acc(a,e,i,\tau,S,S';
      t_{\mathsf{inc}},p_{\mathsf{inc}},
      t_{\mathsf{fin}},p_{\mathsf{fin}}).
\]
The receipt states that the nonce, epoch, fork, expiry,
authorization, and application guards held atomically when the transition
executed.  Later use of the nonce does not invalidate the earlier fact.
We omit timing, position, and state fields when they are immaterial.

\begin{definition}[Ledger authenticator]\label{def:la}
A ledger authenticator relative to a fixed transition machine $\mathcal L$ is
a tuple
\[
 \mathsf{LA}=(\Setup,\Init,\Authorize,\Judge)
\]
with the following algorithms and protocol.
\begin{enumerate}[leftmargin=1.6em]
  \item $\Setup(1^\kappa)$ returns public parameters $\mathsf{pp}$.
  \item $\Init(\mathsf{pp},a)$ returns private state $\mathsf{st}_a$ and a
  public registration event.
  \item $\Authorize(\mathsf{st}_a,b)$ is a reactive classical protocol.  It may emit finitely many ledger events, read finalized prefixes and receipts between emissions, update private state, and eventually fix a canonical action $\tau$.
  \item
  $\Judge(\mathsf{pp},\mathcal T^{\mathsf{fin}},a,\tau)$ is a deterministic
  polynomial time public predicate. Every true judgment is established by a
  finalized receipt for an accepted transition whose canonical action is
  $\tau$, and the earliest such receipt in canonical transcript order is its
  canonical witness.
\end{enumerate}
The predicate has historical semantics, for good finalized prefixes,
\[
 \begin{aligned}
  \Judge(\mathsf{pp},\mathcal T,a,\tau)=1
  \ \wedge\ \mathcal T\preceq\mathcal T'
  \quad\Longrightarrow\quad{}\\
  \Judge(\mathsf{pp},\mathcal T',a,\tau)=1.
 \end{aligned}
\]
\end{definition}

The private state requirement does not make the ledger trusted with a signing secret.  It lets the honest algorithm remember seeds, unused indices, pending processes, and receipts.  In contrast, $\Judge$ uses only public parameters and the finalized transcript.

\subsection{Transcript unforgeability}

Existing unforgeability notions treat authentication as a local predicate over a public key, message, and signature. Ledger authorization instead depends on an evolving public history. We therefore model authentication as a game over the execution of the ledger rather than over isolated signing queries.

The $\LAEUF$ experiment is multi-user. The challenger maintains honest
accounts $\mathsf{Hon}$, private states, reactive authorization processes,
sets $\Requested[a]$, corruption steps $\nu_{\mathsf{cor}}[a]$, and the ledger
transcript. The experiment assigns each interface call, state transition, and
receipt finalization a unique increasing step $\nu$. For every provisional
receipt $R$, let $\nu_{\mathsf{acc}}(R)$ be the unique step at which
$\mathsf{Apply}$ accepts its transition. If $R$ later finalizes, let
$\nu_{\mathsf{fin}}(R)$ be its finalization step. These step values are
analysis metadata maintained by the challenger and are not serialized receipt
fields. A batch finalization assigns its receipts consecutive metadata steps
in canonical transcript order. Let $\mathcal T_\nu^{\mathsf{fin}}$ and
$\Requested_\nu[a]$ denote the finalized prefix and requested set after step
$\nu$. Slots remain the units for deadlines and liveness. The adversary is
quantum polynomial time, but all ledger interfaces are classical.

\begin{definition}[$\LAEUF$ experiment]\label{def:laeuf}
The experiment provides the following interfaces.
Authorization and corruption calls naming an uninitialized account return
$\bot$ without changing any game state.

\paragraph{Honest initialization.}
On $\mathcal O_{\mathsf{init}}(a)$, the challenger rejects a duplicate account
identifier, runs $\Init$, creates an honest account, and installs its
registration in a finalized prefix. The game therefore begins after successful
registration for that account. It returns the public registration record, adds
$a$ to $\mathsf{Hon}$, and sets $\nu_{\mathsf{cor}}[a]=\infty$.

\paragraph{Authorization.}
On $\mathcal O_{\mathsf{auth}}(a,b)$, the challenger starts an honest
$\Authorize(\mathsf{st}_a,b)$ process.  Once the process has selected every
authorization state output $\rho$, it fixes $\tau=\Can(\cdots,b,\rho)$ and adds
$\tau$ to $\Requested[a]$ before exposing the first event for that request.
Every emitted event is delivered to the adversary before possible inclusion.
Concurrent requests are permitted to the extent allowed by the scheme's honest
state machine.  Serialization is a scheme property, not a game assumption.

\paragraph{Corruption.}
On the first call $\mathcal O_{\mathsf{cor}}(a)$ at step $\nu$, the challenger
sets $\nu_{\mathsf{cor}}[a]=\nu$, returns the current private state including
pending process state, and terminates future honest computation for that
account. Already emitted events remain in the pending pool and may still be
scheduled. Later authorization and corruption calls for the account return
$\bot$ without changing $\Requested[a]$.

\paragraph{Scheduling.}
The adversary emits its own events and chooses inclusion and ordering of all pending events.  It may advance the ledger and trigger finalization subject to the ledger transition relation and finality experiment.  It may censor honest events indefinitely.

\paragraph{Winning condition.}
For each pair $(a,\tau)$, define
\[
 \nu_{\mathsf{fin}}^\star(a,\tau)=
 \min\{\nu:
 \Judge(\mathsf{pp},\mathcal T_\nu^{\mathsf{fin}},a,\tau)=1\},
\]
with $\nu_{\mathsf{fin}}^\star(a,\tau)=\infty$ if the set is empty. When it is
finite, let $R^\star(a,\tau)$ be the canonical receipt whose finalization
establishes that first true judgment. The pair is an eligible forgery target if
$\nu_{\mathsf{fin}}^\star(a,\tau)<\infty$ and, writing
$\nu_{\mathsf{acc}}^\star=\nu_{\mathsf{acc}}(R^\star(a,\tau))$,
\[
 a\in\mathsf{Hon},\qquad
 \nu_{\mathsf{cor}}[a]>\nu_{\mathsf{acc}}^\star,\qquad
 \tau\notin\Requested_{\nu_{\mathsf{acc}}^\star}[a].
\]
If any eligible target exists, select the lexicographically first triple
$\bigl(\nu_{\mathsf{fin}}^\star(a,\tau),a,\tau\bigr)$. The adversary wins
exactly when such a triple exists. Its probability is
$\Adv_{\mathsf{LA},\mathcal L}^{\LAEUF}(\mathcal A)$.
\end{definition}

This is canonical action freshness, not encoding freshness.  Encoding a
requested action differently is not a forgery.  Conversely, changing the
successor key or another state effect creates a fresh action even if the
application call is unchanged. We use the first accepting execution step to
test freshness and corruption, and finalization only to decide whether the
historical receipt enters the winning set. Thus a request or corruption after
acceptance does not retroactively change that receipt's eligibility. The
core syntax models a request that fixes one complete action before its first
emission that depends on the action.  A protocol that intentionally authorizes
a set, or chooses the final call after a setup event independent of the action,
can replace $\Requested[a]$ by a public requested action relation.  Its
freshness theorem must then state that relation explicitly.  Treating an
underspecified intent as one message would recreate the canonicalization
problem.

\begin{definition}[Safety despite censorship]\label{def:cis}
A ledger authenticator has safety despite censorship if its $\LAEUF$ advantage
is negligible without an honest inclusion assumption.  The notion may rely on
finality safety and computational assumptions.
\end{definition}

The adjective does not promise information theoretic safety or progress.  It means that indefinite suppression can halt the account but cannot redirect an authorization to an unrequested effect.

\begin{definition}[Conditional completeness]\label{def:complete}
An authenticator is $(k,L)$ complete under $(\Delta_{\mathsf{inc}},\Delta_{\mathsf{fin}})$ if every admissible honest request emits at most $k$ events and produces a finalized accepting receipt within $L$ slots, except with probability $\varepsilon_{\mathsf{live}}+\varepsilon_{\mathsf{fin}}$. Any request arrival window or expiry precondition must be stated explicitly.
\end{definition}

\subsection{Conservative embedding of signatures}

When the ledger provides only transport and replay protection, this definition coincides with the standard authentication notion.  Let $\Sigma=(\mathsf{Kg},\mathsf{Sign},\mathsf{Vf})$ be a signature scheme secure in the usual multi-user post-quantum EUF-CMA experiment.  Construct $\mathsf{LA}_\Sigma$ by registering $\mathsf{vk}_a$, signing the complete canonical action, and accepting an event $(a,\tau,\sigma)$ exactly when $\mathsf{Vf}(\mathsf{vk}_a,\tau,\sigma)=1$ and all state guards hold.

\begin{theorem}[Signature embedding]\label{thm:sig-embed}
For every $\LAEUF$ adversary $\mathcal A$ there is a multi-user $\PQEUF$ adversary $\mathcal B$ such that
\[
 \Adv_{\mathsf{LA}_\Sigma}^{\LAEUF}(\mathcal A)
 \le
 \Adv_{\Sigma}^{\mathsf{mu\text{-}\PQEUF}}(\mathcal B)
 +\varepsilon_{\mathsf{fin}}
 +\varepsilon_{\mathsf{can}}
 +\varepsilon_{\mathsf{state}}.
\]
Here $\varepsilon_{\mathsf{state}}$ bounds any deviation from the specified
transition relation once an event is selected for execution, including invalid
acceptance or failure to create a receipt for a valid transition.  It is zero
for the ideal transition relation.
\end{theorem}

\begin{proof}
The reduction registers the verification keys from its multi-user signature experiment.  It answers an authorization request for canonical action $\tau$ by forwarding $\tau$ to the appropriate signing oracle, then exposes the resulting ledger event.  Scheduling requires no signing secret.

Condition on monotone finality, injective canonicalization, and sound state
execution. At the acceptance step of a winning receipt, its action is
unrequested and its honest verification key is uncorrupted. The event therefore
contains a valid signature on a message absent from the signing queries. The
reduction outputs at the first such transition. A union bound over the three
excluded ledger events proves the claim.
\end{proof}

The theorem is intentionally unsurprising, so a more detailed proof is deferred to Appendix \ref{app:structural-proofs}.  It is a sanity check that the new notion contains the standard portable design point without weakening its security.

\section{Structural boundaries}\label{sec:boundaries}

Theorem~\ref{thm:sig-embed} tells us that Definition~\ref{def:laeuf} admits ordinary signatures, protocols that use transcripts, and
many degenerate mechanisms.  We next isolate two properties that explain when
the ledger is doing cryptographic work, event locality and eligibility closure.

\subsection{The single event frontier}

Not every ledger authenticator fundamentally relies on ledger history. Some protocols authorize an action by emitting one event whose validity is determined entirely by the registration state and a fixed public context. We isolate this class because, despite their ledger presentation, they behave like ordinary one-time authentication schemes. The following definition formalizes the conditions under which such a protocol admits a local verifier.

A public beacon, anchor slot, or active policy may become available after key
registration.  Let $\mathcal G_{\mathsf{ctx}}$ be the public context process
that, after arbitrary preprocessing that depends on the key, produces a context
$\gamma$ and a fresh registered account snapshot $\mathcal T_0$.  The same
process, including adversarial influence on public randomness, is used in the
ledger and extracted signature experiments.

For any ledger snapshot $\mathcal T$, let $\mathsf{pub}(\mathcal T)$ denote
its public parameters, finalized transcript, and public account state, with all
honest private state omitted.

The following definition isolates this \emph{single event form}, and the subsequent theorem shows how it yields a contextual one-time signature.

\begin{definition}[Single event form for a fixed context]\label{def:local}
A ledger authenticator has single event form for a fixed context if there are
polynomial time algorithms $\Emit$ and $V_{\mathsf{loc}}$ with the following properties.
\begin{enumerate}[label=(L\arabic*),leftmargin=2.5em]
  \item \emph{Action programmability and isolated emission.}  There is an
  efficient public map $\mathsf{Req}$ such that, for every admissible canonical
  action $\tau$, the request
  $b=\mathsf{Req}(\mathsf{pub}(\mathcal T_0),\gamma,\tau)$ causes
  authorization to fix exactly $\tau$.  Starting from
  $(\mathsf{st}_a,\mathcal T_0,\gamma)$, it emits exactly
  \[
    \mathsf{ev}=\Emit(\mathsf{st}_a,\mathcal T_0,\gamma,\tau)
  \]
  before reading any extension of $\mathcal T_0$.

  \item \emph{Local completeness.}  For an honest emission,
  \[
   V_{\mathsf{loc}}(\mathsf{pp},\mathsf{pub}(\mathcal T_0),
      \gamma,a,\tau,\mathsf{ev})=1
  \]
  except with probability $\varepsilon_{\mathsf{loc}}$.

  \item \emph{Censor and embed.}  Suppose the honest event is exposed but
  censored.  For every fresh $\tau'$ and event $\mathsf{ev}'$ accepted by
  $V_{\mathsf{loc}}$, a specified uniform polynomial time scheduler
  $\mathcal S_{\mathsf{emb}}$ produces a finalized extension
  $\mathcal T'$ satisfying
  \[
    \Judge(\mathsf{pp},\mathcal T',a,\tau')=1.
  \]
  The embedding fails with probability at most
  $\varepsilon_{\mathsf{emb}}$.  Its elapsed time, sequential depth, and
  oracle use are bounded by a public overhead vector
  $\Delta_{\mathsf{emb}}$ specified by the property.

  \item \emph{Context matching.}  The public registration state and
  $\gamma$ in the signature and ledger experiments have statistical distance,
  or computational distinguishing advantage, at most
  $\varepsilon_{\mathsf{ctx}}$, as appropriate to the reduction.
\end{enumerate}
\end{definition}

Condition (L3) is the substantive locality premise.  A short light client summary does not imply it as unrelated adversarial events can change a summary without increasing its length.  Likewise, a context that expires after a beacon window cannot silently be treated as a permanent signature context.

Let $\mathfrak R$ be an online resource class, such as bounds on elapsed time,
sequential depth, and quantum random oracle queries after the signing event is
revealed.  A \emph{contextual one-time signature} allows arbitrary public key
preprocessing, then generates $\gamma$, answers one signing query, and requires
a forgery on a different canonical action within $\mathfrak R$.

\begin{theorem}[Single event extraction]\label{thm:extract}
Let $\mathsf{LA}$ have single event form for a fixed context and safety despite
censorship.  Define key generation by $\Init$, signing by $\Emit$, and
verification by $V_{\mathsf{loc}}$.  For every contextual one-time forger
$\mathcal F$, there is a straight line ledger adversary
$\mathcal B$ such that
\[
 \Adv^{\mathsf{cOTS}}(\mathcal F)
 \le
 \Adv_{\mathsf{LA}}^{\LAEUF}(\mathcal B)
 +\varepsilon_{\mathsf{loc}}
 +\varepsilon_{\mathsf{ctx}}
 +\varepsilon_{\mathsf{emb}}.
\]
The reduction preserves $\mathfrak R$ up to fixed public validation work and
the stated $\Delta_{\mathsf{emb}}$ overhead, performs no rewinding, and remains
valid in the QROM.
\end{theorem}

\begin{proof}
First replace the signature experiment's public registration and context view
by the ledger distribution.  Condition (L4) changes the success probability by
at most $\varepsilon_{\mathsf{ctx}}$, by coupling in the statistical case and
by a direct distinguishing reduction in the computational case.  In the
matched game, the ledger adversary initializes one honest account and runs the
identical context process.  When $\mathcal F$ requests a signature on $\tau$, it submits
$b=\mathsf{Req}(\mathsf{pub}(\mathcal T_0),\gamma,\tau)$ to the corresponding
authorization.  By (L1), that process fixes $\tau$, and its one honest event is
exposed before any later transcript read.  The adversary forwards that event as
the signature and censors it. If $\mathcal F$ outputs a locally valid
$(\tau',\mathsf{ev}')$ with $\tau'\ne\tau$, condition (L3) supplies a finalized
extension on which $\Judge=1$ for $\tau'$.  Only $\tau$ was recorded in
$\Requested[a]$, so this wins $\LAEUF$.  The context game hop, local failure,
and embedding failure give the three additive terms.  The simulation copies the
forger's quantum computation without measurement or rewinding.
\end{proof}

Within the conditions of Definition~\ref{def:local},
Theorem~\ref{thm:extract} identifies when ledger authentication collapses to
conventional one-time authentication. The converse is tight within that class:
any contextual one-time signature whose public context remains valid for the
required embedding window can be viewed as a single event ledger authenticator.

\begin{proposition}[Converse]\label{prop:ots-converse}
Every contextual one-time signature whose public context remains admissible on
the ledger for the stated $\Delta_{\mathsf{emb}}$ window after signing induces
a single event ledger authenticator for a fixed context.  Register its
verification key, emit $(a,\gamma,\tau,\sigma)$, and let the ledger check the
context, state guards, and signature before creating a historical receipt.
\end{proposition}

\begin{corollary}[Single event equivalence]\label{cor:one-event}
Up to context, canonicalization, finality, and state execution errors,
single event ledger authentication for a fixed context is black-box equivalent to
contextual one-time signatures in the same online resource class.  If
$\gamma$ is empty or fixed at key generation and validity does not expire, the extracted object is an ordinary one-time signature.
\end{corollary}

Encrypted mempools, trusted ordering, proof systems, and substantial sequential
work can change the resource class or violate censor and embed.  Schemes using
several events escape because the verifier uses a prior finalized event
unavailable to a self-contained signature verifier.

\subsection{Why eligibility must close before revelation}

The previous section identified when ledger authentication reduces to ordinary one-time authentication. We now consider the second boundary, where a credential is intentionally revealed after an earlier commitment. The question is no longer whether authorization is local, but whether finalized history permanently fixes every future use of the disclosed credential. The following definition isolates this proof boundary.


\begin{definition}[Rebindable reveal protocol]\label{def:rebindable}
Fix a live account cell $u$.  A rebindable reveal protocol has the following
properties.
\begin{enumerate}[label=(R\arabic*),leftmargin=2.5em]
  \item a public predicate
  $\mathsf{Eligible}(\mathcal T,u,c)$ for precursor evidence $c$.
  \item an honest final event $r=(u,\tau,s,w)$ that exposes a credential $s$
  and an opening witness $w$, with $\mathsf{CredOK}(u,s)=1$.
  \item public predicates $\mathsf{CredOK}(u,s)$ and
  $\mathsf{Open}(c,u,\tau,s,w)$ such that, while $u$ is live, a final event is
  accepted and creates a receipt whenever its canonical guards hold and some
  eligible $c$ satisfies both predicates.
  \item an efficient algorithm
  \[
    \mathsf{Rebind}(u,s,\tau')\longrightarrow(c',w')
  \]
  that, for any other admissible action $\tau'$, except with probability
  $\varepsilon_{\mathsf{rb}}$, satisfies
  \[
   \mathsf{Open}(c',u,\tau',s,w')=1.
  \]
\end{enumerate}
The live cell has at least two distinct admissible canonical actions.  It
remains live until a final event is accepted.  Precursor evidence alone does
not consume it.
\end{definition}

This class includes a credential based on a hash preimage with commitments of
the form $H(\tau,s,h^+,\cdots)$.  Once $s$ is public, anyone can select a fresh
action, successor head, and any public opening randomness, then compute matching
evidence.

\begin{definition}[Admission after disclosure]\label{def:post}
Fix an efficient scheduler $\mathcal P$.  After an honest final event
$r=(u,\tau,s,w)$ is emitted but before it is included, $\mathcal P$ censors
$r$, chooses a distinct admissible action $\tau'$, and computes
\[
 (c',w')\leftarrow\mathsf{Rebind}(u,s,\tau').
\]
The admission event after disclosure occurs if $\mathcal P$ produces a finalized
extension in which this same $c'$ is eligible while $u$ is live and schedules
$r'=(u,\tau',s,w')$ for inclusion before either condition expires.  Write
$p_{\mathsf{post}}(\mathcal P)$ for its probability.
\end{definition}

\begin{theorem}[Admission attack after disclosure]\label{thm:closure}
For every rebindable reveal protocol and scheduler after disclosure
$\mathcal P$, there is a $\LAEUF$ adversary $\mathcal A_{\mathcal P}$ such that
\[
 \Adv_{\mathsf{LA}}^{\LAEUF}(\mathcal A_{\mathcal P})
 \ge
 p_{\mathsf{post}}(\mathcal P)
 -\varepsilon_{\mathsf{rb}}
 -\varepsilon_{\mathsf{fin}}
 -\varepsilon_{\mathsf{state}}.
\]
In particular, if evidence constructed after disclosure is always admissible,
the protocol lacks safety despite censorship.
\end{theorem}

\begin{proof}
The adversary requests $\tau$, exposes and censors its honest final event, and
runs $\mathcal P$.  On the admission event after disclosure, the same evidence
$c'$ returned with $w'$ is eligible when $\mathcal P$ schedules
$r'=(u,\tau',s,w')$.  The credential and eligibility checks hold, and the cell
remains live.  The rebinding check fails with probability at most
$\varepsilon_{\mathsf{rb}}$.  Outside finality or failure of state execution,
$r'$ creates a finalized historical receipt for $\tau'$, while only $\tau$
belongs to $\Requested[a]$.  A union bound gives the claim.
\end{proof}

The theorem explains why finalizing one commitment is insufficient: every commitment that may later authorize using the disclosed credential must already be fixed.

\begin{corollary}[Eligibility closure]\label{cor:closure}
Within the rebindable reveal class, and without an honest inclusion assumption,
every efficient scheduler's probability of making evidence constructed after
disclosure usable must be negligible. When admission uses only public
computation and ledger scheduling, this condition is enforced by requiring
every eligible precursor for the live credential to be irrevocably represented
by finalized state before the first event exposing that credential.
\end{corollary}

A finalized accumulator root can close a large set.  Closure does not require
a literal list.  Merely placing $s$ inside a commitment is insufficient if a
new commitment can still be admitted after $s$ leaks.  The corollary is
necessary only for rebindable reveal protocols.  A one-time signature that
cannot be rebound, trusted ordering, a confidential pending pool, or an
independent computational admission gate escapes this implementation-level
conclusion.

\section{Ledger executions in the QROM}\label{sec:lqrom}

A ledger protocol is not a non-interactive QROM algorithm.  It is an execution in which a quantum adversary communicates through classical mempool, ledger, authorization, and corruption interfaces, while the ledger and honest parties evaluate the same random oracle.  We specify that joint execution before using a finalized transcript as a proof cut.

\subsection{Execution semantics and query accounting}

Let $H:\mathcal X\rightarrow\mathcal Y$ be uniformly random with $|\mathcal Y|=2^\lambda$.

\begin{definition}[Ledger QROM execution]\label{def:lqrom}
The experiment $\LQROM_{\mathcal E,\mathcal A}^{H}(1^\kappa)$ consists of a classical ledger environment $\mathcal E$, a quantum adversary $\mathcal A$, and the common oracle $H$.

Between classical calls, $\mathcal A$ may apply an arbitrary quantum circuit and query
\[
 |x,z,w\rangle\longmapsto|x,z\oplus H(x),w\rangle.
\]
To call a ledger interface, it designates a message register.  The environment
measures that register in the computational basis, processes the resulting
classical command, and returns a classical response.  Only the interface
register is measured.  All remaining adversarial workspace persists across
slots, blocks, finality events, and calls.

The environment may evaluate the same $H$ on classical inputs while validating events or producing honest messages.  It updates the pending pool, state, clock, transcript, and finalized prefix, and exposes every public emission and receipt classically.
\end{definition}

The definition does not give coherent access to the state transition, mempool, or finality predicate.  A ledger accepting superpositions of transactions would be a different primitive. We distinguish three sources of oracle evaluations.  Direct quantum queries by the adversary are counted by $\Qadv$.  Verifier mediated evaluations on
adversarially supplied inputs whose result affects a visible response are counted by $\Qmed$.  Honest evaluations on private inputs are counted by $\Qhon$ in the joint oracle state, although an instance proof may show that some cannot contribute to the relevant search relation.  A conservative bound uses
\[
 \Qeff=\Qadv+\Qmed+\Qhon+Q_{\mathsf{out}},
\]
where $Q_{\mathsf{out}}$ includes retained candidate checks, including checks
at the first accepting transition.

Charging $\Qmed$ is essential.  Even an accept/reject result can provide a
classical test that depends on the oracle.  A classical evaluation is safely simulated
as one quantum query.  Omitting it lets rejected ledger events create an
uncounted search interface.  Conversely, finality does not reset the
adversary's state or query budget.

\subsection{Adaptive finalized cuts}

Eligibility closure ensures that the set of admissible targets eventually becomes fixed. We next define a closed finalized cut, the point at which this set becomes immutable and the remaining execution is analyzed under a fixed target quantum search bound. Let $\mathcal F_t$ be the classical filtration generated by the public
transcript, responses, clock, randomness, and finality certificates through slot $t$.

\begin{definition}[Closed finalized cut]\label{def:cut}
A finalized cut is an $\mathcal F_t$ stopping time $\Theta$. At the cut, the
environment computes
\[
 C_\Theta=g(\mathcal T^{\mathsf{fin}}_\Theta)
 \subseteq\mathcal Y
\]
and copies it into a read only classical register $C$. The cut is
\emph{closed for the acceptance relation} if, on every continuation of a cut
branch, no value outside $C_\Theta$ can become eligible for a later winning
check. Equivalently, whenever a later acceptance is witnessed by a checked
point $x$, its target value satisfies $H(x)\in C_\Theta$. Thus closure
constrains the eligibility relation used by later verification, not merely the
copied register. The adversary may influence $\Theta$ and $C_\Theta$ through
scheduling, oracle queries, and events before the cut, but every branch must
satisfy the public bound $|C_\Theta|\le K$.
\end{definition}

When slots contain several ordered transition phases, the cut definition must
name the phase at which the copy is taken.  A close at the end of slot $t$
includes every finalization processed in that slot and precedes every event of
slot $t+1$.  This convention rules out an otherwise ambiguous admission race
within the same slot.

The target set need not be independent of $H$.  It may contain honest
commitment digests and oracle outputs selected by the adversary.  The essential
condition is that a classical copy is frozen before the later disclosure or
computation used to open it. Fix a classical cut transcript $z$ and a hidden
value $\omega$ that is sampled and stored in a classical environment register
no later than the cut, although it may be revealed afterward. The predicate
$\Fresh(z,\omega,x)$ is deterministic, oracle independent, and fixed from the
cut onward. Any oracle evaluation needed to determine freshness must instead
be represented in the target relation and charged to the query budget. In the
compressed oracle representation, a database is a
partial function
$D:\mathcal X\rightarrow\mathcal Y\cup\{\bot\}$.  Define
$\Hit_{z,\omega,C}(D)=1$ iff there exists an $x$ satisfying
\[
 x\in\operatorname{dom}(D),\qquad
 \Fresh(z,\omega,x)=1,
 \qquad D(x)\in C.
\]

The same argument applies to any exact oracle purification with a
protocol specific database projector for which every oracle free map
intertwines the projector, every charged oracle evaluation satisfies the stated
single query transition bound, and the retained final checks satisfy the stated
output recording inequality. The proof of Lemma~\ref{lem:cut} uses only these
properties. The CCR application uses this form with a compressed database for
its marked slice oracle.

At the adaptive cut, write the purified state as a direct sum over classical branches,
\[
 |\Psi_\Theta\rangle
 =\bigoplus_{\theta,z,\omega,C}
 |\theta,z,\omega,C\rangle
 |\psi_{\theta,z,\omega,C}\rangle.
\]
Let $\Pi_{\Hit}$ project onto compressed databases satisfying the relation.
The \emph{bad norm at the cut} is
\[
 \beta_{\mathsf{cut}}
 =\|\Pi_{\Hit}|\Psi_\Theta\rangle\|.
\]
This projector is analytic as the experiment does not measure the compressed database or disturb the adversary at finality.

\subsection{The finalized target lemma}

Once the target set is fixed by a closed finalized cut, the remainder of the execution reduces to quantum search over that immutable set.

\begin{lemma}[Finalized target search]\label{lem:cut}
Consider a closed finalized cut with $|C|\le K$, and let $W$ be the event of
the first winning acceptance. Assume
\[
 W\subseteq\{\Theta<\infty\},
\]
and that, on every branch in $W$, the cut occurs before every checked oracle
evaluation used to witness that acceptance. Branches on which the cut does not
occur are assigned $W=0$. Suppose the retained validation history contains one
of at most $\ell$ checked points $x_j$ for which
\[
 \Fresh(z,\omega,x_j)=1,
 \qquad H(x_j)\in C.
\]
For every such check, the verifier evaluates $y_j=H(x_j)$ after the cut,
charges that evaluation to the post-cut budget, and copies $(x_j,y_j)$ into a
classical check register. Let at most $Q$ oracle evaluations occur after the
cut, counting every direct, mediated, honest, and final check evaluation
capable of affecting the relation. There is a universal constant
$c_{\mathsf{co}}$ such that
\[
 \sqrt{\Pr[W]}
 \le
 \beta_{\mathsf{cut}}
 +Q\sqrt{\frac{c_{\mathsf{co}}K}{2^\lambda}}
 +\sqrt{\frac{2\ell}{2^\lambda}}.
\]
Consequently,
\[
 \Pr[W]
 \le
 3\beta_{\mathsf{cut}}^2
 +\frac{3c_{\mathsf{co}}KQ^2}{2^\lambda}
 +\frac{6\ell}{2^\lambda}.
\]
The standard transition capacity formulation permits
$c_{\mathsf{co}}=10$~\cite{chung2021compressed}.
\end{lemma}

We give a sketch of the proof below. The full proof of Lemma \ref{lem:cut} 
is in Appendix \ref{app:cut-proof}.

\begin{proof}[Proof sketch]
Replace the common random oracle with its purified compressed oracle
simulation~\cite{zhandry2019record}.  Decompose the state at the cut into its
$\Hit$ and $\neg\Hit$ components.  The first has norm
$\beta_{\mathsf{cut}}$ and is carried through the bound.  We neither measure it
nor condition on a nonexistent classical transcript of quantum queries.  On
the $\neg\Hit$ component, $C$ is classical and immutable.  For a fresh
unrecorded input, a newly assigned uniform oracle value enters $C$ with
probability at most $K/2^\lambda$.  This remains true under arbitrary
correlation between $C$ and the compressed database before the cut induced by the
classical transcript.  The environment never reads that analytic database
directly, and the frozen value of $C$ is not a function of the new value
assigned at the fresh input.  The transition capacity theorem bounds the
norm transferred from $\neg\Hit$ to $\Hit$ by one evaluation by
$\sqrt{c_{\mathsf{co}}K/2^\lambda}$.  Telescoping over all $Q$ direct and mediated evaluations gives the second amplitude term.

The recording lemma relates the accepted classical check history to its
database entries.  Retaining $\ell$ checked points contributes at most
$2\ell/2^\lambda$ in probability.  The triangle inequality proves the first
display, and $(a+b+c)^2\le3(a^2+b^2+c^2)$ proves the second.  For an adaptive
cut, apply the same uniform bound to each classical direct sum branch and
average, or equivalently pad every branch after a classical cut flag to the
maximum query budget.
\end{proof}

The lemma permits $C$ to contain $H(x_0)$ for inputs before the cut.  A known
authorized $x_0$ does not satisfy $\Fresh$.  Finding a different fresh input is
a fixed target preimage or second preimage problem.  If a fresh $x_0$ already
hits $C$, it belongs to the bad component at the cut and must be charged through
$\beta_{\mathsf{cut}}$.

This does not make generic hash commitments secure against preimage attacks.  An adversary may
find a collision before the cut, commit to one side, and retain the other.  A
  concrete protocol obtains the $Q^2/2^\lambda$ target bound only after proving
that every precomputed fresh alternative is included in a negligible bad norm
at the cut.  Otherwise QROM collision attacks remain possible. 

The previous lemma bounds the probability of a successful fresh target search following a single closed finalized cut. Summing these contributions over all closed cuts in an execution yields a bound on the probability of any fresh hit occurring during the account lifetime.

\begin{corollary}[Lifetime target accounting]\label{cor:lifetime}
Suppose protected episodes $j$ have closed target sets of sizes $K_j$ and
probability contributions $\varepsilon_{\mathsf{pre},j}$ from the bad component
at the cut.  Let
\[
 K_{\mathsf{life}}=\sum_j K_j,
 \qquad
 \varepsilon_{\mathsf{pre}}^{\mathsf{life}}
 =\sum_j\varepsilon_{\mathsf{pre},j},
\]
and let $Q_{\mathsf{life}}$ bound every relevant lifetime evaluation.  Then
\[
 \Pr[\mathsf{win}]
 \le
 3\varepsilon_{\mathsf{pre}}^{\mathsf{life}}
 +\frac{3c_{\mathsf{co}}K_{\mathsf{life}}
 Q_{\mathsf{life}}^2}{2^\lambda}
 +\frac{6\ell_{\mathsf{life}}}{2^\lambda}.
\]
\end{corollary}

\section{Worked instantiation of commit, close, reveal}\label{sec:ccr}

We now present a concrete instantiation of the framework using a recurring hash-based authenticator. The purpose of this construction is not to introduce an entirely new cryptographic primitive, but to identify the precise properties required for security by this design family and to show that it satisfies them within our framework.

We call the construction \emph{commit, close, reveal} ($\CR$).  The
construction is a state machine formulation of the Guy Fawkes idea and is
included to exercise the definitions and finalized target lemma. Commitment, waiting for finality, eligibility closure, and the conversion of theft into denial of service all have clear antecedents in Fawkescoin~\cite{bonneau2014fawkescoin}.

\subsection{Two tempting failures}

Closure is not sufficient when the precursor fails to bind the credential. Suppose a commitment contains an action and successor head but omits $s_i$. Before the close, an adversary chooses its own fresh action, plants the matching commitment, and waits.  After the honest pending reveal discloses $s_i$, it censors that reveal and combines the stolen credential with its already eligible action.  No random oracle inversion is needed.  In the language of
Section~\ref{sec:lqrom}, the protocol has failed to make the bad norm at the cut
small.

The opposite error is to bind $s_i$ but leave admission open.  The adversary
then learns $s_i$, computes a new commitment for its action, finalizes that
commitment while censoring the honest reveal, and opens it.  This is exactly
Theorem~\ref{thm:closure}.  The construction below therefore needs both
properties.  The precursor includes the hidden credential, and every eligible
precursor is finalized before disclosure.

\subsection{Construction}

Let $F:\bits^\kappa\times\bits^*\rightarrow\bits^{\lambda_s}$ be a PRF secure
against quantum adversaries~\cite{zhandry2012qprf}. Let $H_0$ and $H_1$ be
separate independent quantum random oracle interfaces, with output lengths
$\lambda_h$ and $\lambda_c$. Fix a public injective prefix free encoder
$\mathsf{Enc}$. It begins with a fixed suite and type tag, uses fixed width
encodings for bounded integers, and length delimits every variable length
field. If one master oracle implements both interfaces, the interface and type
tags give disjoint domains. The ledger rejects noncanonical encodings,
including noncanonical $u,d$, and $\tau$ fields.

Fix public caps $M,N_{\mathsf{cell}}\ge1$ and define
\[
 \begin{aligned}
 \mathsf{params}=(&\mathsf{suiteId},\mathsf{version},\mathsf{encId},
 \kappa,\lambda_s,\lambda_h,\lambda_c,\lambda_r,\\
 &M,D_{\mathsf{com}},N_{\mathsf{cell}},
 \mathsf{canId},\mathsf{finalityId}).
 \end{aligned}
\]
The initialization of account $a$ fixes its registration label $e$, samples
$k_a\leftarrow\bits^\kappa$ independently, and stores it in
$\mathsf{st}_a$. Distinct cells have distinct contexts and hence distinct PRF
inputs. For cell $u=(a,e,i)$, define
\[
 \ctx_u=\mathsf{Enc}(\mathtt{ctx};
 \mathsf{chainId},\mathsf{forkId},a,e,i,\mathsf{params})
\]
and
\[
 \begin{aligned}
 s_i&=F_{k_a}\bigl(\mathsf{Enc}(\mathtt{secret};\ctx_u)\bigr),\\
 h_i&=H_0\bigl(\mathsf{Enc}(\mathtt{head};\ctx_u,s_i)\bigr).
 \end{aligned}
\]
Registration derives $s_0,h_0$ this way and stores $a,e,h_0$, the parameter
tuple, and $N_{\mathsf{cell}}$ in finalized account state. The ledger requires
\[
 s\in\bits^{\lambda_s},\quad r\in\bits^{\lambda_r},\quad
 c\in\bits^{\lambda_c},\quad h\in\bits^{\lambda_h}.
\]
For fixed $u$, every well-formed commitment input has a unique factorization
\[
 \mathsf{Enc}(\mathtt{commit};\ctx_u,d,\tau,s,r)
 =\mathsf{Enc}_u(\varphi,s),
\]
where $\varphi$ contains every field other than the secret coordinate. This is
the factorization used in Lemma~\ref{lem:head-exposure}.

When registration, or the accepted reveal for cell $i-1$, first belongs to a
finalized prefix at slot $t_i^{\mathsf{open}}$, the ledger opens cell $u$ and
sets the deadline slot for the commitment window of cell $i$ as
\[
 d_i=t_i^{\mathsf{open}}+D_{\mathsf{com}},
\]
where $D_{\mathsf{com}}$ dictates how long the commitment phase remains open. Thus every cell receives a full commitment interval.  Starting the clock from
mere inclusion would let a later reorganization shorten or move the interval.

A commitment digest $c$ is an admission candidate for $u$ when a commit event
$(\mathtt{commit},u,d_i,c)$ first appears in a finalized prefix at a slot
$t_C$ satisfying
\[
 t_i^{\mathsf{open}}<t_C\le d_i.
\]
Order distinct candidates by finalization slot, their canonical position in the
finalized transcript, and then digest.  Let $C_u$ contain the first $M$
candidates in this order, or every candidate if fewer than $M$ appear.  These
and only these digests are eligible.  The ledger copies $C_u$ after processing
all finalizations in slot $d_i$.  It then freezes, and events finalized in
later slots can never enter it.  Thus
\begin{equation}
 |C_u|\le M,
 \qquad
\sum_u|C_u|\le M N_{\mathsf{life}}.
\label{eq:admission-cap}
\end{equation}
Here $N_{\mathsf{life}}$ denotes a public lifetime cap on protected cells.  The
full security and resource caps are fixed below.
The cap is a safety parameter rather than a liveness guarantee.  An adversary
may fill it before the honest digest is finalized and park the cell, but cannot
turn that exclusion into authorization of a fresh action.

To authorize application body $b$, the signer derives
\[
 u^+=(a,e,i+1),\qquad
 h_{i+1}=H_0\bigl(\mathsf{Enc}(\mathtt{head};\ctx_{u^+},s_{i+1})\bigr)
\]
and forms the complete canonical action
\[
 \tau=\Can(\mathsf{chainId},\mathsf{forkId},a,e,i,
           b,h_{i+1},d_i,\mathsf{params}).
\]
At this point the $\LAEUF$ authorization interface fixes and records $\tau$ in
$\Requested[a]$.  The signer then samples
$r_i\leftarrow\bits^{\lambda_r}$ and computes
\[
 c_i=H_1\bigl(
 \mathsf{Enc}(\mathtt{commit};\ctx_u,d_i,\tau,s_i,r_i)
 \bigr).
\]
It emits $(\mathtt{commit},u,d_i,c_i)$.

The signer waits until this digest belongs to $C_u$.  If it misses the close,
the signer never reveals $s_i$.  Otherwise it additionally waits until the
close has passed and emits
\[
 (\mathtt{reveal},u,\tau,s_i,r_i).
\]
The ledger accepts that event only when all of the following conditions hold:
\begin{enumerate}[leftmargin=1.5em]
  \item $u$ is the current live cell and the canonical action guards hold.
  \item the reveal is included after $d_i$.
  \item $H_0\bigl(\mathsf{Enc}(\mathtt{head};\ctx_u,s_i)\bigr)=h_i$.
  \item
  $H_1\bigl(\mathsf{Enc}(\mathtt{commit};
  \ctx_u,d_i,\tau,s_i,r_i)\bigr)\in C_u$.
\end{enumerate}
It executes $b$, installs the $h_{i+1}$ already contained in $\tau$, and
creates a historical receipt.  Only the first accepted reveal can consume the
cell.  When that receipt finalizes, its finalization slot becomes
$t_{i+1}^{\mathsf{open}}$.

The construction is the ledger authenticator
\[
 \CR=(\Setup,\Init,\Authorize,\Judge)
\]
relative to the fixed transition machine $\mathcal L_{\CR}$ implementing the
rules above. $\Setup$ returns $\mathsf{pp}=\mathsf{params}$.
$\Init(\mathsf{pp},a)$ samples $k_a$, derives
$h_0$, initializes private state
\[
 (k_a,e,0,\mathsf{pending}=\bot,\mathsf{mode}=\mathtt{live}),
\]
and returns the canonical registration event.

The honest state machine serializes authorization. $\Authorize$ accepts a
request only if $i<N_{\mathsf{cell}}$, the current cell is live, its close has
not been processed, and $\mathsf{pending}=\bot$. It fixes $\tau$ and atomically stores
\[
 \mathsf{pending}=(u,\tau,s_i,r_i,c_i,h_{i+1})
\]
before its first emission. A refused concurrent request returns $\bot$ before
fixing an action. If the close is observed with $\mathsf{pending}=\bot$, the
process parks the account. If $c_i\notin C_u$ at the close, it parks the account
without revealing $s_i$. If $c_i\in C_u$, it emits the reveal above and remains
pending while that event may be censored. When the receipt finalizes, the
reactive process observes it, advances the private state, and clears
$\mathsf{pending}$. The public transition has already advanced the ledger
state. After $N_{\mathsf{cell}}$ accepted cells, the account is exhausted and
no further request is admitted. Finally,
\[
 \Judge(\mathsf{pp},\mathcal T^{\mathsf{fin}},a,\tau)=1
\]
exactly when $\mathcal T^{\mathsf{fin}}$ contains a finalized
$\mathcal L_{\CR}$ receipt for an accepted reveal of $(a,\tau)$. Thus a
censored commitment parks the account rather than exposing its secret, and a
censored reveal cannot open a new admission interval under the exposed head.

Let $\delta_{\mathsf{slot}}\in\{0,1\}$ be the delay caused by rounding from the
close to the first strictly later inclusion slot.

\begin{proposition}[Completeness under bounded contention]\label{prop:ccr-complete}
Suppose an account submits an admissible serialized request at
$t_{\mathsf{req}}$ and remains uncorrupted through the displayed deadline.
Thus its current cell is live, $i<N_{\mathsf{cell}}$, and
no honest authorization process is pending. Suppose
\[
 t_i^{\mathsf{open}}<t_{\mathsf{req}},
 \qquad
 t_{\mathsf{req}}+\Delta_{\mathsf{inc}}
 +\Delta_{\mathsf{fin}}<d_i.
\]
Suppose also that fewer than $M$ distinct candidates precede the honest digest
in the canonical admission order, and that the canonical and application
guards for the requested action remain valid through the reveal deadline.
Under $(\Delta_{\mathsf{inc}},\Delta_{\mathsf{fin}})$ liveness, outside
liveness, finality, or state-execution failure, either the requested action
produces a finalized receipt by
\[
 d_i+\delta_{\mathsf{slot}}
 +\Delta_{\mathsf{inc}}+\Delta_{\mathsf{fin}}
\]
or a fresh finalized receipt for the same cell appears first. Consequently,
letting $\mathsf{Miss}$ denote failure of this dichotomy by the displayed
bound,
\[
 \Pr[\mathsf{Miss}]
 \le \varepsilon_{\mathsf{live}}+\varepsilon_{\mathsf{fin}}
 +\varepsilon_{\mathsf{state}}+\Adv_{\CR}^{\LAEUF}(\mathcal A).
\]
Outside a liveness, finality, or state-execution failure and a $\LAEUF$ win,
the requested action therefore has a finalized receipt by the displayed bound.
The construction emits at most two events, and the elapsed bound is
$L\le D_{\mathsf{com}}+\delta_{\mathsf{slot}}+
\Delta_{\mathsf{inc}}+\Delta_{\mathsf{fin}}$. Without the admission-rank
premise, the scheme provides safety but not completeness.
\end{proposition}

\subsection{Ledger QROM security}

The transition machine admits only $0\le i<N_{\mathsf{cell}}$. Hence
\[
 N_{\mathsf{life}}\le U_{\mathsf{life}}N_{\mathsf{cell}},
 \qquad
 K_{\mathsf{life}}\le M U_{\mathsf{life}}N_{\mathsf{cell}}.
\]
The cell and target caps are therefore protocol enforced. The oracle and
retained check caps below remain adversarial resource bounds.

Fix public deterministic caps $U_{\mathsf{life}},Q_0,Q_1$.  For each account
$a$, fix public quotas $N_a\le N_{\mathsf{cell}},K_a,\ell_a$ on its
protocol admissible protected cells, total eligible targets, and pairs retained
from opening checks up to the first fresh acceptance.
Set
\[
 N_{\mathsf{life}}=\sum_aN_a,\qquad
 K_{\mathsf{life}}=\sum_aK_a,\qquad
 \ell_{\mathsf{life}}=\sum_a\ell_a.
\]
In every execution under consideration, at most $U_{\mathsf{life}}$ honest
accounts are initialized, account $a$ respects its three quotas, and at most
$Q_b$ evaluations of $H_b$ occur for $b\in\{0,1\}$.  The query caps include
every direct, mediated, relevant honest, and evaluation used to check outputs.
These quotas indexed by account support the hybrid for the target account
below.  With only an undifferentiated global cap, one may set every quota for an
account to that cap,
which exposes the corresponding $U_{\mathsf{life}}$ factor rather than hiding
it.
By~\eqref{eq:admission-cap}, one may take $K_a=M N_a$ and hence
$K_{\mathsf{life}}=M N_{\mathsf{life}}$.  Put
\[
 \Qhead:=Q_0+Q_1.
\]
For every initialized honest account $a$, let $\mathcal B_a$ be the $\qPRFQRO$
hybrid reduction that uses its challenge for account $a$ and stops if $a$ is
corrupted before the first unrequested accepting transition under that account.
Define the aggregate
multi-user loss
\[
\Delta_{\mathsf{PRF}}
 :=\sum_{a=1}^{U_{\mathsf{life}}}
 \Adv_F^{\qPRFQRO}(\mathcal B_a).
\]
Here $\qPRFQRO$ is the PRF experiment in which the distinguisher also has
quantum access to independent random oracles $H_0$ and $H_1$, independent of
its challenge oracle. Under a multi-key $\qPRFQRO$ convention,
$\Delta_{\mathsf{PRF}}$ is the corresponding single multi-key advantage.
Each $\mathcal B_a$ makes at most $N_a+1$ challenge oracle derivations, inherits
the caps $Q_0,Q_1$ on its auxiliary QRO calls, and runs in the time of
$\mathcal A$ plus polynomial simulation overhead.

\begin{lemma}[CCR exposure through both oracles]
\label{lem:head-exposure}
Work in the hybrid with random secrets.  Let at most $N$ protected cells have
independent secrets
\[
 s_u\leftarrow\bits^{\lambda_s},
 \qquad
 h_u=H_0\bigl(\mathsf{Enc}(\mathtt{head};\ctx_u,s_u)\bigr),
\]
where the head prefixes are distinct and $H_0,H_1$ are independent random
oracles.  For each cell $u$, let $\Theta_u$ be its close.  Before $\Theta_u$,
assume that the public values depending on $s_u$ consist only of the public
head, outputs of $H_1$ on prescribed honest commitment inputs, and responses
produced by charged $H_0$ and $H_1$ evaluations.  There is no other leakage or
uncharged predicate involving $s_u$.  Prescribed input encodings are distinct
and efficiently recognizable by the simulator from their fields other than a
parsed candidate coordinate for the secret.  Let
$\mathcal R_u$ be the value of $\Requested[a]$ at $\Theta_u$.  Every prescribed
honest commitment action belongs to $\mathcal R_u$.

Use the routed simulation constructed below, and let $D_{P,u}$ be the standard
compressed database of its marked slice oracle $P_u$. For a cut branch $z$,
define
\[
 B_{u,z}(D_{P,u})=1
\]
if there is a $\varphi\in\operatorname{dom}(D_{P,u})$ that is a well formed
opening description for $u$ and whose extracted action lies outside the frozen
set $\mathcal R_u$. Let $\Pi_u^P$ be the corresponding branch controlled
projector and write
\[
 \widehat\beta_u
 =\|\Pi_u^P|\Psi_{\Theta_u}\rangle\|.
\]
Let $R$ bound the complete sequence of charged calls that can carry a candidate
secret. This includes every direct, mediated, honest, and final check
evaluation of $H_0$ or $H_1$ that can carry such a coordinate. Then there are
universal constants $c_{\mathsf{eq}},c_{\mathsf{alt}}>0$ such that
\begin{align}
 \sum_u\widehat\beta_u^2
 &\le
 c_{\mathsf{eq}}
 \frac{N(R+1)^2}{2^{\lambda_s}},
 \label{eq:head-exposure-eq}\\
 \Pr[\mathsf{AltHead}_R]
 &\le
 c_{\mathsf{alt}}
 \frac{N(R+1)^2}{2^{\lambda_h}},
 \label{eq:head-exposure-alt}
\end{align}
where $\mathsf{AltHead}_R$ is the event that a retained or output candidate
$(u,s)$ with $s\ne s_u$ is classically checked within the charged sequence and
satisfies
$H_0\bigl(\mathsf{Enc}(\mathtt{head};\ctx_u,s)\bigr)=h_u$. The final check is
included. The second bound remains valid after disclosure of the planted
secrets. It bounds a found and checked alternative, not the
information theoretic existence of other preimages.
\end{lemma}

\begin{proof}
Fix a cell $u$. Route its planted head slice and its well formed commitment
slice through an equality oracle
$E_u(s)=\mathbf 1[s=s_u]$. On the marked slice, $H_0$ returns the sampled
public head and $H_1$ calls an independent random oracle $P_u$ on the unique
secret deleted description $\varphi$; off the slice, independent random
oracles provide the answers. For fixed $s_u$, this routing is a basis
permutation of independent truth tables, so the visible interfaces are exactly
independent random oracles. Prescribed honest calls and routed adversarial
calls use the same $P_u$, including when their inputs overlap.

Compare the routed circuit with the circuit in which $E_u$ is replaced by the
all zero oracle. In the latter, every entry of the standard compressed
database $D_{P,u}$ comes from a prescribed honest commitment and therefore has
an action in $\mathcal R_u$. Hence $\Pi_u^P$ kills the zero oracle state. The
BBBV hybrid bound~\cite{bennett1997strengths}, applied to the at most $2R$
equality calls needed to compute and uncompute the routing bit, gives
\[
 \widehat\beta_u^2
 \le O\left(\frac{(R+1)^2}{2^{\lambda_s}}\right).
\]
Retaining the adaptive close flag and padding each branch to the public cap
makes this an ordinary fixed length hybrid. Summing over cells proves
\eqref{eq:head-exposure-eq}.

After the cut, only a controlled $P_u$ call can create a fresh database entry
whose value lies in the frozen set $C_u$. Its transition norm is at most
$\sqrt{c_{\mathsf{co}}|C_u|/2^{\lambda_c}}$; prescribed honest calls have zero
fresh transition capacity. An exact secret final check retains the tested
input and output, so the output recording bound applies. Finally, outside the
planted point, the values of $H_0$ remain independent and uniform even after
$s_u$ is disclosed. Standard quantum search, including the final classical
check, therefore proves~\eqref{eq:head-exposure-alt}. Full routing and
compressed database details appear in Appendix~\ref{app:ccr-exposure}.
\end{proof}

\begin{theorem}[$\LQROM$ security of $\CR$]\label{thm:ccr}
For every $\LAEUF$ adversary $\mathcal A$, there are $\qPRFQRO$ adversaries
$\{\mathcal B_a\}_a$ and universal constants $c_h,c_c$ such that
\begin{align*}
 \Adv_{\CR}^{\LAEUF}(\mathcal A)
 \le{}&
 \Delta_{\mathsf{PRF}}
 +\varepsilon_{\mathsf{fin}}
 +\varepsilon_{\mathsf{can}}
 +\varepsilon_{\mathsf{state}} \\
 &+c_h\frac{N_{\mathsf{life}}(\Qhead+1)^2}
 {2^{\min\{\lambda_s,\lambda_h\}}} \\
 &+c_c\frac{K_{\mathsf{life}}Q_1^2}
 {2^{\lambda_c}}
 +\frac{6\ell_{\mathsf{life}}}{2^{\lambda_c}}.
\end{align*}
One may take $c_c=3c_{\mathsf{co}}$.  The safety statement assumes no
honest inclusion bound.
\end{theorem}

Full proofs for the $\LQROM$ security of the commit, close, reveal (CCR) 
authenticator are given in Appendix \ref{app:ccr-proofs}.
\begin{proof}[Proof sketch of Theorem~\ref{thm:ccr}]
Let $E_{\mathsf{acc}}$ be the event that an uncorrupted honest account has an
unrequested accepting transition, and let $W_a$ select the account of the first
such transition. Every $\LAEUF$ win implies $E_{\mathsf{acc}}$. For each $a$,
use a separate unconditioned hybrid replacing that account's PRF outputs by
independent uniform strings up to its corruption. Bound $W_a$ without
conditioning the state or random oracles on that event, and then sum over $a$.
Abort on finality, canonicalization, or state execution failure. The public
per-account quotas sum to the displayed lifetime aggregates. Fix the selected
transition and its cell $u$.

At the close $d_i$, the finalized transcript fixes $C_u$, and every later
acceptance relevant eligible digest lies in that set. Define a fresh opening
input to contain the context $u$, a canonical action outside the frozen set $\mathcal R_u$, a secret satisfying the registered head, and an unauthorized commitment opening.

Split the selected accepting reveal according to whether its secret is the planted
$s_u$.  If it uses $s\ne s_u$, then
\[
 H_0\bigl(\mathsf{Enc}(\mathtt{head};\ctx_u,s)\bigr)=h_i
\]
is an alternative head preimage.  Equation~\eqref{eq:head-exposure-alt} with
$R=\Qhead$ bounds this case by
\[
 c_{\mathsf{alt}}
 \frac{N_{\mathsf{life}}(\Qhead+1)^2}{2^{\lambda_h}}.
\]
For openings using the exact secret, use the marked slice purification
from Lemma~\ref{lem:head-exposure} and let $\beta_u$ be its bad norm at the cut
for the finalized target. Its projector additionally requires the $H_1$ output to lie in $C_u$, so it is
dominated by $\Pi_u^P$ in Lemma~\ref{lem:head-exposure}. Hence
\[
 \sum_u\beta_u^2
 \le c_{\mathsf{eq}}
 \frac{N_{\mathsf{life}}(\Qhead+1)^2}{2^{\lambda_s}}.
\]
This joint accounting covers both a blind $H_1$ carrier query and a search
assisted by $H_0$ followed by an $H_1$ carrier query.

On the complementary subspace, the adversary may retain all quantum state from
before the cut and may later learn $s_i$ from the honest pending reveal. It
cannot make a digest outside $C_u$ eligible. A fresh accepted action therefore
supplies a new structured
input whose $H_1$ image lies in one of the lifetime finalized target sets.
Lemma~\ref{lem:cut} and Corollary~\ref{cor:lifetime} give the commitment term
and the term for checking outputs.  Adversarial ordering may select which valid reveal runs
first, but it cannot create a target after closure.  Censorship may stop every
reveal and therefore affects completeness, not this case split.  Combining the
alternative case and the case using the exact secret gives the theorem with
$c_h=c_{\mathsf{alt}}+3c_{\mathsf{eq}}$ and
$c_c=3c_{\mathsf{co}}$.
\end{proof}

\begin{corollary}[Capped lifetime targets]\label{cor:ccr-cap}
For every execution containing at most $N_{\mathsf{life}}$ protected cells,
\[
 \sum_u|C_u|\le M N_{\mathsf{life}}.
\]
Thus one may take $K_{\mathsf{life}}=M N_{\mathsf{life}}$ in
Theorem~\ref{thm:ccr}, and its finalized target contribution is at most
\[
 c_c\frac{M N_{\mathsf{life}}Q_1^2}{2^{\lambda_c}}.
\]
Filling all $M$ positions may park the cell, but does not enlarge the
accepting target set after disclosure.
\end{corollary}

The theorem explains why the commitment term admits a target preimage bound
only in this construction.  Before disclosure, every fresh valid opening also
inverts the registered head and is charged to the bad term at the cut.  After
disclosure, all eligible commitment outputs are fixed.  Without that first
case, a generic collision found before the cut would require collision
analysis.

\begin{corollary}[Lifetime parameter sizing]\label{cor:ccr-parameters}
Write $\lambda_H=\min\{\lambda_s,\lambda_h\}$ and
\[
 \Delta_0=\Delta_{\mathsf{PRF}}+\varepsilon_{\mathsf{fin}}
 +\varepsilon_{\mathsf{can}}+\varepsilon_{\mathsf{state}}.
\]
Assume the deterministic caps
\[
\begin{aligned}
 N_{\mathsf{life}}&\le N,
 &K_{\mathsf{life}}&\le MN,\\
 \Qhead&\le q_h,
 &Q_1&\le q_c,
 &\ell_{\mathsf{life}}&\le L.
\end{aligned}
\]
For any $\epsilon_h,\epsilon_c>0$, it is sufficient to choose
\begin{align*}
 \lambda_H&\ge
 \left\lceil\log_2\!\left(
   \frac{c_hN(q_h+1)^2}{\epsilon_h}
 \right)\right\rceil,\\
 \lambda_c&\ge
 \left\lceil\log_2\!\left(
   \frac{c_cMNq_c^2+6L}{\epsilon_c}
 \right)\right\rceil
\end{align*}
to obtain
\[
 \Adv_{\CR}^{\LAEUF}(\mathcal A)
 \le \Delta_0+\epsilon_h+\epsilon_c.
\]
\end{corollary}

\begin{proof}
Substitute the public caps and $K_{\mathsf{life}}\le MN$ into
Theorem~\ref{thm:ccr}.
\end{proof}

The PRF key length $\kappa$ and suite must separately make
$\Delta_{\mathsf{PRF}}$ smaller than the deployment's PRF failure budget under
the stated user and query caps. The value $\lambda_r$ controls commitment
randomization but contributes no independent security exponent to the theorem;
it may be zero when deterministic commitments are acceptable.

Ignoring universal constants and the residual from recording outputs, a $w$-bit
quantum work target requires approximately
\[
 \lambda_H\ge 2w+\log_2N,
 \qquad
 \lambda_c\ge 2w+\log_2K_{\mathsf{life}}.
\]

\section{Design and proof implications}\label{sec:implications}

The framework exposes design choices that are hidden when every mechanism is
called a ``signature.'' Table~\ref{tab:design-map} summarizes the main points
on the spectrum. Additional design and deployment detail appears in
Appendix~\ref{app:deployment-detail}.

\paragraph{Finality is part of the cryptographic API.}

A protocol must wait under the same finality predicate used in its theorem.
Revealing after mere inclusion may expose a credential on a fork whose
precursor disappears. The finality rule, chain and fork identifiers, close
duration, and canonicalization policy therefore belong to the authentication
domain. The additive $\varepsilon_{\mathsf{fin}}$ permits composition with a
concrete consensus bound without assuming deterministic irreversibility.

\paragraph{Admission policy affects the cryptographic bound.}

Before an opening, a commitment that appears random is normally
unauthenticated. Fees, deposits, authenticated fee payers, per-cell caps, and
accumulators therefore affect the cryptographic bound by determining $K$ and
the verifier's query budget. A full cap may deny progress, but it still
supports safety when every admitted target is fixed before disclosure and
charged in the loss.

\paragraph{Portability is the exchanged resource.}

A portable signature can be verified from a small tuple outside the ledger.
A transcript authenticator instead requires finalized history or authenticated
proofs for its precursor and receipt. It exchanges portability and latency for
evidence maintained by the consensus system.

\paragraph{Limitations and deployment considerations.}\label{sec:limitations}

The $\LQROM$ keeps ledger calls classical, and finality does not erase quantum
advice. Finalized registration, independent recovery authority, and serialized
wallet state are deployment premises; Appendix~\ref{app:deployment-detail}
gives the full limitations.

\section{Conclusion}\label{sec:conclusion}

Ledger authorization can be transcript dependent, finalized commitments,
nonce history, and event order affect verification. $\LAEUF$ models scheduling
and censorship while separating safety from progress. Our single event conditions yield a contextual one-time signature. In rebindable reveals with otherwise public admission, safety requires finalized closure before disclosure. In the QROM, finality is a classical proof cut, not a reset; quantum state persists and the loss scales with the target set. Independent gates rely on their own assumptions.

\appendix
\section*{APPENDIX}
\begin{table*}[t]
\centering
\scriptsize
\setlength{\tabcolsep}{3.2pt}
\caption{Ledger authentication design points.  Events count authorization
events after finalized registration.  Authentication material excludes the
canonical action and common transaction fields.  Sizes are symbolic because
encodings and parameter sets differ.}
\label{tab:design-map}
\begin{tabularx}{\textwidth}{
 @{}l c
 >{\raggedright\arraybackslash}X
 >{\raggedright\arraybackslash}X
 c
 >{\raggedright\arraybackslash}X@{}}
\toprule
Mechanism & Events & Material for each action
& Public state and validation & Portable
& Censorship and proof status \\
\midrule
ML-DSA or SLH-DSA
& \(1\)
& Signature \(\sigma_\Sigma\)
& Registered verification key, replay state, and one
  \(\mathsf{Vf}_\Sigma\) invocation
& Yes
& Censorship affects progress, not unforgeability under multi-user
  \(\PQEUF\) \cite{nistfips204,nistfips205}. \\

XMSS or LMS
& \(1\)
& Stateful hash-based signature, including its index and authentication data
& Registered root, replay state, and tree plus hash checks for a one-time signature
& Yes
& Safety also requires persistent index discipline at the signer
  \cite{hulsing2018xmss,mcgrew2019lms,nist800208}. \\

Fawkescoin
& Several
& Hash commitment, disclosure, and transfer data
& Finalized timeline and UTXO state with hash checks
& No
& Rules for earliest transfer and expiry convert the disclosure race from theft
  to freezing \cite{bonneau2014fawkescoin}. \\

HBAuth
& Depends on the batch
& Hash authorization associated with a metatransaction batch
& Contract state, hash checks, and rules favoring the preferred batcher
& No
& Censorship and admission depend on the batcher model.  No ledger QROM
  theorem is claimed here \cite{hughes2023hbauth}. \\

Finlow-Bates et al.
& \(2\)
& Fixed-length hash fields in commitment and reveal transactions
& Confirmed commitment and hash verification
& No
& Admission after disclosure must close for safety despite censorship.
  Formal analysis is left open in the proposal \cite{finlow2026cost}. \\

\(\CR\) (this work)
& \(2\)
& \(c_i\), reveal values \((s_i,r_i)\), and successor head \(h_{i+1}\)
& \(h_i,d_i,C_u\), with \(|C_u|\le M\), plus one \(H_0\) and one \(H_1\)
  reveal check
& No
& A full cap may park the cell.  Theorem~\ref{thm:ccr} gives lifetime
  \(\LQROM\) security under \(\LAEUF\) without an inclusion assumption. \\
\bottomrule
\end{tabularx}
\end{table*}

\section{Additional design and deployment detail}
\label{app:deployment-detail}

\paragraph{Finality as part of the API.}
A protocol must bind its waiting rule to the same finality predicate used in
the theorem. Revealing after mere inclusion can expose a credential on a fork
whose commitment is later removed. On the surviving history the adversary may
then create the first eligible evidence. Changing the finality rule, chain or
fork identifier, close duration, or canonicalization policy changes the
authentication domain and should therefore change $\ctx$. The additive finality term is useful for composition. A ledger authenticator proof can be reused with a bound specific to the consensus protocol on
$\varepsilon_{\mathsf{fin}}$, rather than silently assuming deterministic
irreversibility. Long-range rewrites remain outside a short range finality
claim and require whatever checkpointing or weak subjectivity assumption the
underlying ledger already needs.

\paragraph{Admission policy.}
Before an opening, a commitment that appears random is normally
unauthenticated. An attacker can therefore enlarge an eligible target set or
trigger validation unless the ledger charges another resource. Fees, deposits,
independently authenticated fee payers, caps for each cell, and storage using
an accumulator are systems mechanisms, but they have cryptographic
consequences as they determine the realized $K$ and the query budget for the
verifier in the finalized target lemma. A cap is not a liveness guarantee. An attacker that fills it can exclude the
honest commitment. It can nevertheless support safety where every admitted
target is fixed before disclosure and charged in the loss over all targets.
Conversely, calling validation ``free'' in a QROM proof while allowing
unlimited rejected events gives the adversary an unaccounted interface that
depends on the oracle.

\paragraph{Portability.}
A portable signature can be verified from a small tuple outside the ledger. A
transcript authenticator may require the finalized history or authenticated
proofs for both the precursor and final receipt. A bridge, court, offline
device, or another chain must import those facts before using $\Judge$. The
ledger construction has not compressed a signature in the ordinary sense. It
has exchanged portability and latency for evidence already maintained by the
consensus system.

\paragraph{Parameter interpretation.}
A 256 bit commitment output and $K_{\mathsf{life}}=2^{40}$ give a nominal
lifetime search exponent for the commitment of 108, not 128. The construction
carries a commitment, a revealed secret, a successor head inside the canonical
action, and optional opening randomness. Its attraction is verification using
only hashes. Its cost is latency from two events, bounded transcript state, and
evidence that is not portable.

\paragraph{Model and deployment limits.}
The $\LQROM$ models quantum computation and quantum access to the designated
hash functions, while events, receipts, and scheduler calls remain classical.
This captures the setting in which the cryptographic primitive may be queried
in superposition, but the ledger itself exposes a classical execution history.
A coherently queryable state transition or a ledger accepting superpositions
of transactions requires a different model. The finalized target rule also
requires a classical immutable target register and an explicit bad amplitude
at the cut. Finality does not erase arbitrary advice derived from oracle calls,
and the rule does not cover target sets that remain updateable after
disclosure. Consequently, the security argument applies only after the
eligible target set has become immutable; protocols that continue modifying
eligibility after credential disclosure fall outside this model. The experiment begins from finalized registration. Genesis, governance, an
existing credential, or a separately authenticated recovery path must install
the initial state. Safety despite censorship also permits a cell to freeze.
This separates safety from liveness, the model prevents unauthorized state
transitions even when an adversary can indefinitely delay honest progress.
Reopening a commitment interval after its credential may have leaked can
violate closure, so recovery needs independent registered authority, such as a
post-quantum recovery key, guardian policy, or precommitted epoch rotation. A
ledger nonce prevents two actions from both being accepted, but does not
prevent a signer from exposing two pending authenticators derived from
one-time state. Concrete wallets may therefore need serialization and
crash-consistent state management. These implementation requirements are
outside the ledger security game but are necessary for a real deployment to
preserve the one-time-use assumptions. Finally, the compact construction uses
random oracles. Alternatives in the standard model have different efficiency
and proof profiles.

\section{Full reductions for the structural results}
\label{app:structural-proofs}
This appendix makes explicit the games, stopping rules, and corruption
handling suppressed in the main text.  Write
\[
 \mathsf{Good}_{\mathsf{led}}
 =\neg\Bad_{\mathsf{fin}}\wedge
  \neg\Bad_{\mathsf{can}}\wedge
  \neg\Bad_{\mathsf{state}}.
\]
The three bad events in the experiment are that finalized prefixes cease to be
monotone, canonicalization ceases to be effect injective, or state execution
deviates from the specified transition relation after an event is selected.
Thus
\begin{equation}
 \Pr[\neg\mathsf{Good}_{\mathsf{led}}]
 \le \varepsilon_{\mathsf{fin}}
    +\varepsilon_{\mathsf{can}}
    +\varepsilon_{\mathsf{state}}.
 \label{eq:ledger-bad-union}
\end{equation}
All reductions stop at the first unrequested accepting transition relevant to
their output. Consequently, a corruption at a later game step neither
invalidates the ledger win nor has to be answered by the reduction.

\subsection{Signature embedding}

For clarity, the multi-user signature experiment used in
Theorem~\ref{thm:sig-embed} has classical signing and corruption interfaces. It
returns a verification key for every honest user, answers signing queries, and,
on corruption, returns that user's signing key and removes the user from the
set of permissible forgery targets.  A forgery is a valid signature on a
message never submitted to the corresponding signing oracle, under a user
uncorrupted when the forgery is returned.  This is the multi-user, adaptive
corruption formulation matching Definition~\ref{def:laeuf}.  If one starts
instead from a single-user definition without corruption, the usual target
user guess gives the corresponding factor in the number of initialized honest
accounts.

\begin{proof}[Full proof of Theorem~\ref{thm:sig-embed}]
Reduction \(\mathcal B\) runs \(\mathcal A\) and the ledger experiment. On \(\mathcal O_{\mathsf{init}}(a)\), it obtains a fresh verification key \(\mathsf{vk}_a\) from its multi-user challenger, forms the prescribed registration event, and installs that registration in the initial finalized prefix.  All account and ledger state other than signing keys is public or is sampled by \(\mathcal B\) exactly as in the real experiment. For an authorization request, the honest state machine first selects every authorization state output \(\rho\), including a successor credential, and forms
\[
 \tau=\Can(\mathsf{chainId},\mathsf{forkId},a,e,i,b,\rho).
\]
The reduction inserts \(\tau\) into \(\Requested[a]\), queries the signing oracle for user \(a\) on exactly \(\tau\), and exposes \((a,\tau,\sigma)\) to \(\mathcal A\).  Notice that the signing query and the update of \(\Requested[a]\) occur before exposure, as required by Definition~\ref{def:laeuf}.  The reduction answers a corruption query for \(a\) by forwarding it to the signature challenger and returning the signing key together with the locally maintained account and pending process state. It processes all adversarial events, scheduling decisions, receipts, and finalization commands by the public ledger transition relation, none of these steps requires a signing key.

Suppose \(\mathcal A\) causes the first accepting receipt
\[
 R^\star=\Acc(a^\star,e^\star,i^\star,\tau^\star,\ldots)
\]
that satisfies the winning condition of Definition~\ref{def:laeuf}.
Write \(\nu_{\mathsf{acc}}^\star=\nu_{\mathsf{acc}}(R^\star)\). Reduction
\(\mathcal B\) stops at the first unrequested accepting transition under an
uncorrupted honest account. On a ledger win, such a transition occurs no later
than \(\nu_{\mathsf{acc}}^\star\). Denote the account, action, and
authentication string at the stopping step by
\((a^\dagger,\tau^\dagger,\sigma^\dagger)\). On
\(\mathsf{Good}_{\mathsf{led}}\), sound state execution implies
\[
 \mathsf{Vf}(\mathsf{vk}_{a^\dagger},\tau^\dagger,
             \sigma^\dagger)=1.
\]
The target account is uncorrupted at the stopping step, so it is an eligible
target in the multi-user signature game. Every message sent to its signing
oracle was inserted into \(\Requested[a^\dagger]\). Since the accepted
\(\tau^\dagger\) is outside that set at the stopping step, it was
never signed. Effect injectivity rules out interpreting a signed canonical
action as the fresh transition. Thus
\((a^\dagger,\tau^\dagger,\sigma^\dagger)\) is a multi-user signature forgery.
We have established the event containment
\[
 \mathsf{Win}_{\LAEUF}\wedge\mathsf{Good}_{\mathsf{led}}
 \subseteq \mathsf{Win}_{\mathsf{mu\text{-}\PQEUF}}.
\]
Taking probabilities and applying~\eqref{eq:ledger-bad-union} gives exactly the bound in Theorem~\ref{thm:sig-embed}.  The simulation is straight line and all interfaces exposed to \(\mathcal A\) are classical, so the reduction also applies to a quantum adversary.
\end{proof}

\subsection{Single event extraction and converse}

We spell out the contextual one-time experiment used in Theorem~\ref{thm:extract}.  Key generation runs \(\Init\) and gives the
forger the public registration view.  After arbitrary public preprocessing that
depends on the key, the same process \(\mathcal G_{\mathsf{ctx}}\) used in the
ledger experiment returns \((\gamma,\mathcal T_0)\).  The forger makes one
classical signing query on an admissible canonical action \(\tau\), receives one
event \(\mathsf{ev}\), and wins by returning \((\tau',\mathsf{ev}')\) within
\(\mathfrak R\) such that
\[
 \tau'\ne\tau
 \quad\text{and}\quad
 V_{\mathsf{loc}}(\mathsf{pp},\mathsf{pub}(\mathcal T_0),
   \gamma,a,\tau',\mathsf{ev}')=1.
\]
The target signing state is never revealed.  Multi-user variants permit
corruption of instances other than the target and use the same first acceptance
rule as \(\LAEUF\).

\begin{proof}[Full proof of Theorem~\ref{thm:extract}]
First replace the contextual one-time experiment's public registration and
context view by the ledger distribution.  By (L4), this changes
\(\mathcal F\)'s success probability by at most
\(\varepsilon_{\mathsf{ctx}}\).  Statistical matching uses a coupling.
Computational matching uses the composed simulation as a distinguisher.
Work in the resulting matched game.  Ledger adversary \(\mathcal B\) creates
one honest account \(a\), obtains its fresh finalized registration snapshot
\(\mathcal T_0\), and runs the public context process. When \(\mathcal F\) asks
to sign \(\tau\), reduction \(\mathcal B\) computes
\(b=\mathsf{Req}(\mathsf{pub}(\mathcal T_0),\gamma,\tau)\) and invokes
\(\mathcal O_{\mathsf{auth}}(a,b)\).  By L1, the honest process fixes exactly
\(\tau\), which the ledger records in \(\Requested[a]\), and before reading any
extension of \(\mathcal T_0\) it emits
\[
 \mathsf{ev}=
 \Emit(\mathsf{st}_a,\mathcal T_0,\gamma,\tau).
\]
This is exactly the contextual signature returned to \(\mathcal F\).
Reduction \(\mathcal B\) exposes it, as the ledger game requires, and then censors it permanently.  By L2, the event is locally valid except with probability \(\varepsilon_{\mathsf{loc}}\). Suppose \(\mathcal F\) outputs a winning \((\tau',\mathsf{ev}')\).  The reduction has made no authorization request other than \(\tau\), never corrupts \(a\), and has kept the honest event out of the ledger.  Since \(\tau'\ne\tau\), condition (L3) applies.  Reduction \(\mathcal B\) runs the public censor and embed scheduler on \((\tau',\mathsf{ev}')\).  Except with probability \(\varepsilon_{\mathsf{emb}}\), it obtains a finalized extension \(\mathcal T'\) with
\[
 \Judge(\mathsf{pp},\mathcal T',a,\tau')=1.
\]
Here \(\varepsilon_{\mathsf{emb}}\) includes every finality or state execution failure in the promised finalized embedding.  The target remains uncorrupted and \(\tau'\notin\Requested[a]\), so the first such receipt is a \(\LAEUF\) win.

Let \(E_{\mathsf{loc}}\) and \(E_{\mathsf{emb}}\) denote the two failures in
the matched game.  Then
\[
 \mathsf{Win}_{\mathsf{cOTS}}
 \wedge\neg(E_{\mathsf{loc}}\vee E_{\mathsf{emb}})
 \subseteq \mathsf{Win}_{\LAEUF}.
\]
The union bound in the matched game, followed by the context game hop, yields
the theorem. The reduction invokes \(\mathcal F\) once,
in order, and never measures or rewinds its private workspace.  The only
additional work is generation and public validation of the same context plus
the scheduler promised by L3. It therefore preserves elapsed time, sequential
depth, and lifetime oracle limits in \(\mathfrak R\), up to the stated public
validation and \(\Delta_{\mathsf{emb}}\) overhead.  In particular, the argument
remains valid when \(\mathcal F\) retains quantum state and makes quantum random
oracle queries.
\end{proof}

\begin{proof}[Full proof of Proposition~\ref{prop:ots-converse}]
Let
\[
 \Sigma_1=(\mathsf{Kg}_1,\mathsf{Sign}_1,\mathsf{Vf}_1)
\]
be a contextual one-time signature. The initializer runs
\[
 (\mathsf{vk},\mathsf{sk})\leftarrow\mathsf{Kg}_1(1^\kappa),
\]
registers \(\mathsf{vk}\), and stores
\((\mathsf{sk},\mathtt{unused})\). For the current public snapshot and context,
authorization of \(\tau\) emits
\[
 \mathsf{ev}=(a,\gamma,\tau,
    \mathsf{Sign}_1(\mathsf{sk},\mathsf{pub}(\mathcal T_0),
                    \gamma,\tau))
\]
and changes the private flag to \(\mathtt{pending}\).  The state machine
refuses a second authorization under that key, even if the first event is censored.  The ledger accepts the event only when the registered key, context, nonce, epoch, expiry, and application guards match and
\[
 \mathsf{Vf}_1(\mathsf{vk},\mathsf{pub}(\mathcal T_0),
               \gamma,\tau,\sigma)=1.
\]
Acceptance consumes the one-time cell and creates the historical receipt used by \(\Judge\).

Property L1 is immediate because signing emits one event without reading a
later transcript.  Taking \(V_{\mathsf{loc}}=\mathsf{Vf}_1\), signature
correctness gives L2.  For L3, censor the honest event.  The public cell and its
one-time verification key then remain unused.  The proposition's premise on the
context window ensures that any locally valid fresh event remains admissible for the
stated embedding schedule.  It can therefore be included and finalized, after
which the receipt makes \(\Judge=1\).  Finality or state failure is absorbed in
\(\varepsilon_{\mathsf{emb}}\).  Using the identical context process gives
(L4) with zero distance.

For completeness, security is preserved as well.  On good ledger execution, a first receipt for a fresh canonical action contains a valid contextual one-time signature different from the at most one requested action for that cell.  Hence a multi-user ledger forger gives a multi-user contextual one-time forger with
\[
 \Adv_{\mathsf{LA}_{\Sigma_1}}^{\LAEUF}
 \le \Adv_{\Sigma_1}^{\mathsf{mu\text{-}cOTS}}
 +\varepsilon_{\mathsf{fin}}
 +\varepsilon_{\mathsf{can}}
 +\varepsilon_{\mathsf{state}}.
\]
If only a single user cOTS definition is assumed and at most \(q_{\mathsf{init}}\) one-time cells are registered, guessing the target cell gives the standard factor \(q_{\mathsf{init}}\) on the cOTS term. Adaptive corruption is handled exactly as in the signature embedding proof: the reduction outputs at the first unrequested accepting transition, before corruption of the target cell.
\end{proof}

\subsection{Admission after disclosure}

\begin{proof}[Full proof of Theorem~\ref{thm:closure}]
Fix an honest live cell \(u\).  Adversary \(\mathcal A_{\mathcal P}\) makes
exactly one authorization request for an action \(\tau\), so
\(\Requested[a]=\{\tau\}\) for the target cell.  It permits any honest precursor
needed by the protocol to be processed.  When the final honest event
\[
 r=(u,\tau,s,w)
\]
is emitted into the public pending pool, the adversary learns \(s\) and
censors \(r\) before inclusion.  In particular, the cell remains live.
It now runs \(\mathcal P\), which chooses a complete admissible canonical action
\(\tau'\ne\tau\), including all effects on the successor credential and
authorization state, and computes
\((c',w')\leftarrow\mathsf{Rebind}(u,s,\tau')\).  The adversary never asks the
authorization oracle for \(\tau'\).

Let \(E_{\mathsf{post}}\) be the event of Definition~\ref{def:post}.  By
definition,
\(\Pr[E_{\mathsf{post}}]=p_{\mathsf{post}}(\mathcal P)\).  On this event,
\(\mathcal P\) has produced a finalized extension in which this same \(c'\) is
eligible while \(u\) is live and has scheduled
\[
 r'=(u,\tau',s,w').
\]
The credential check succeeds because \(s\) is the credential emitted by the
honest process.  Except with probability \(\varepsilon_{\mathsf{rb}}\), the
rebinding guarantee gives
\(\mathsf{Open}(c',u,\tau',s,w')=1\).  Thus all acceptance predicates in
Definition~\ref{def:rebindable} hold.  Sound state execution accepts \(r'\),
executes \(\tau'\), and
creates a historical receipt. The adversary schedules that receipt to
finality.  No assumption of honest inclusion or liveness is used, all events whose
progress is needed are adversarially scheduled, while the honest reveal
remains censored. The adversary never corrupts the target account, so its
corruption step is infinite. At the acceptance step for \(r'\), the account is
honest and \(\tau'\notin\Requested[a]\); its receipt later finalizes. Hence this
receipt satisfies every condition of Definition~\ref{def:laeuf}.  Formally,
\[
 E_{\mathsf{post}}
 \wedge\neg E_{\mathsf{rb}}
 \wedge\neg\Bad_{\mathsf{fin}}
 \wedge\neg\Bad_{\mathsf{state}}
 \subseteq \mathsf{Win}_{\LAEUF},
\]
where \(\Pr[E_{\mathsf{rb}}]\le\varepsilon_{\mathsf{rb}}\).  Taking
probabilities and applying the union bound proves
\[
 \Adv_{\mathsf{LA}}^{\LAEUF}(\mathcal A_{\mathcal P})
 \ge p_{\mathsf{post}}(\mathcal P)-\varepsilon_{\mathsf{rb}}
      -\varepsilon_{\mathsf{fin}}-\varepsilon_{\mathsf{state}}.
\]
There is no canonicalization term because the attack chooses two distinct canonical actions, rather than relying on two encodings of one action.
\end{proof}

\section{Full proof of the finalized target lemma}
\label{app:cut-proof}

This appendix proves Lemma~\ref{lem:cut} and Corollary~\ref{cor:lifetime}.  The proof is self-contained as a reduction to two standard facts about compressed oracles, the transition capacity bound for one query and the output recording lemma of Zhandry and Chung et al.~\cite{zhandry2019record,chung2021compressed}.  We state those facts in the precise form used below.  We do not reproduce the construction and soundness proof of the compressed oracle itself. We first make explicit one condition implicit in the direct sum display in Section~\ref{sec:lqrom}.  The value $\omega$ may be hidden from the adversary at the cut and disclosed later, but it is sampled and stored in a classical environment register no later than the cut.  Thus, on every cut branch, the predicate
\[
  x\longmapsto \Fresh(z,\omega,x)
\]
is fixed from the cut onward.  The lemma does not cover a value $\omega$ that is first generated from oracle answers after the cut, nor a freshness predicate that changes with later ledger state.  In either case the target relation itself could change without an oracle query and the transition argument below would not apply.

\subsection{Purifying the ledger execution}
\label{app:purified-execution}

Fix an execution $\LQROM_{\mathcal E,\mathcal A}^{H}(1^\kappa)$ as in Definition~\ref{def:lqrom}.  Since the environment and adversary are polynomial-time, their registers and the execution horizon are finite for each $\kappa$.  We purify the execution as follows. Every randomized choice is generated coherently from a fresh coin register. Every CPTP map is replaced by a Stinespring isometry, retaining its environment register.  When a message, response, receipt, clock value, or finality certificate is declared classical, its value in the computational basis is copied to an inaccessible history register.  Tracing out that register has exactly the effect of the measurement or dephasing in the original execution.  Keeping it instead yields a global pure state whose different classical histories occupy orthogonal summands.  This purification does not give the adversary coherent ledger access, the adversary receives only the same dephased classical response as in Definition~\ref{def:lqrom}. An environment evaluation of $H$ on a classical input is implemented by one ordinary oracle call followed by a copy of its answer in the computational basis. This covers honest evaluations and evaluations made by the verifier.  We split every environment transition around these calls, so between two consecutive oracle calls the global isometry does not act on the oracle purification. Consequently, all direct, mediated, honest, and final evaluations occur in one ordered list.  This justifies charging them together in $\Qeff$.

\subsection{Compressed oracle representation}
\label{app:compressed-oracle}

Let $N=|\mathcal Y|=2^\lambda$.  Replace the uniformly random function oracle by its purified compressed oracle simulation \cite{zhandry2019record,chung2021compressed}.  Its database register has an orthonormal basis indexed by finite partial functions
\[
  D:\mathcal X\longrightarrow\mathcal Y\cup\{\bot\}.
\]
We write $\operatorname{dom}(D)=\{x:D(x)\ne\bot\}$.  The replacement is exact
on all registers visible to $\mathcal E$ and $\mathcal A$.  Only a
compressed oracle call acts on the database register.  In particular, local
adversarial computation and ledger processing that makes no oracle evaluation
cannot add an input and output pair to $D$.  We use the following two standard
consequences of the compressed oracle analysis.  They are stated for sequential
calls.  A classical oracle evaluation is a special case of such a call.

\begin{lemma}[Single query transition capacity]
\label{lem:app-transition}
Let $R$ be a predicate on compressed databases and let $\Pi_R$ be its projector.  Suppose that for every $D$ with $R(D)=0$ and every possible query input $x$, all range values whose assignment at $x$ can make $R$ true belong to a set $L_{D,x}\subseteq\mathcal Y$ of size at most $r$.  If membership in $R$ is witnessed by the newly present IO pair, then for one compressed oracle call $O_{\mathsf{co}}$,
\[
 \bigl\|\Pi_R O_{\mathsf{co}}(I-\Pi_R)\bigr\|_{\mathrm{op}}
 \le
 \sqrt{\frac{c_{\mathsf{co}}r}{N}}.
\]
The transition capacity normalization used in Chung et al.~\cite{chung2021compressed} permits $c_{\mathsf{co}}=10$.
\end{lemma}

\begin{lemma}[Output recording]
\label{lem:app-recording}
Suppose a classical validation process retains at most $\ell$ checked points
$x_1,\ldots,x_\ell$ before it stops. These may include attempts made through
the verifier and points in a final output. The verifier evaluates
$y_j=H(x_j)$ for every checked point, counts each evaluation in the oracle
budget, and copies $(x_j,y_j)$ into a classical check register.  Let
$|\Phi\rangle$ be the purified state after these checks.  Let $W$ be an event
whose occurrence implies that at least one retained pair $(x_j,y_j)$ satisfies
a fixed relation $R_{\mathsf{out}}$.  Conditioned on the orthogonal classical
check branches, let $\Pi_{R_{\mathsf{db}}}$ project onto compressed databases
that contain the exact retained pair $(x_j,y_j)$ for at least one satisfying
check.  Then
\[
 \sqrt{\Pr[W]}
 \le
 \bigl\|\Pi_{R_{\mathsf{db}}}|\Phi\rangle\bigr\|
 +\sqrt{\frac{2\ell}{N}}.
\]
The statement remains true when the points are chosen adaptively and when their
classical descriptions are correlated with the preceding execution.
\begin{proof}
Condition on an orthogonal classical validation branch and delete repeated
checked points. The standard output-recording inequality applies to the at
most $\ell$ remaining exact pairs. Its bound is uniform over the branch, so the
subnormalized branch vectors recombine in Euclidean norm. Enlarging the usual
residual to the displayed $\sqrt{2\ell/N}$ gives the claim.
\end{proof}
\end{lemma}

The qualifications concerning the exact pair and checks made afterward are essential.  The probability
that a newly evaluated $y_j$ lies in a target set of size $K$ is charged by
that final oracle evaluation through the transition bound.  Recording then
loses only $O(1/N)$ per exact retained pair.  Without charged final evaluations
and retained responses, a relation such as $H(x)\in C$ would instead incur a
residual proportional to $|C|/N$.

Lemma~\ref{lem:app-transition} is the strongly recognizable transition bound,
and Lemma~\ref{lem:app-recording} is the usual link between the compressed and
standard oracles.  Their constants depend on the convention used to normalize
the compressed oracle.  Keeping $c_{\mathsf{co}}$ symbolic makes the proof
independent of that convention.  The displayed value $10$ is inherited from
the cited transition capacity theorem rather than rederived here.

\subsection{The state at an adaptive cut}
\label{app:cut-state}

We implement the stopping time $\Theta$ by a classical cut flag $F$.  Before the flag is set, the environment tests at each slot whether the stopping condition has occurred.  Because $\Theta$ is an $\mathcal F_t$ stopping time, this test uses only the classical history through the current slot. When it first succeeds, the environment computes
\[
 C=g(\mathcal T^{\mathsf{fin}}_\Theta)
\]
and coherently copies the classical tuple $(\Theta,z,\omega,C)$ into a read-only cut register and its purifying history register.  The source registers are classical, so this copy is well defined. Closure means that every later acceptance-relevant target remains contained in this copied set. Let
\[
 b=(\theta,z,\omega,C)
\]
denote a cut branch.  Retaining the inaccessible history registers gives the orthogonal decomposition
\begin{equation}
 |\Psi_\Theta\rangle
 =\bigoplus_b |b\rangle|\psi_b\rangle,
 \label{eq:app-cut-direct-sum}
\end{equation}
where the vectors $|\psi_b\rangle$ are subnormalized conditional states at
their branch-specific stopping times. Adjoin a cemetery branch $b=\dagger$ for
executions in which the cut does not occur within the experiment horizon, set
$\Pi_\dagger=0$, and declare the winning event false on that branch. Then
$\sum_b\|\psi_b\|^2=1$. This stopped classical-quantum ensemble does not assert
independence between $C$ and the database. All correlation caused by oracle
queries before the cut, honest hash computations, adversarially planted
targets, scheduling, and finalization remains inside the branch states
$|\psi_b\rangle$. For a non-cemetery branch $b$, define
\[
 \Fresh_b(x)=\Fresh(z,\omega,x)
\]
and set $R_b(D)=1$ exactly when there exists an $x$ satisfying
\[
 x\in\operatorname{dom}(D),\qquad
 \Fresh_b(x)=1,\qquad D(x)\in C.
\]
Let $\Pi_b$ project onto the databases satisfying $R_b$, tensored with the identity on every register other than the database.  The global hit projector is
\begin{equation}
 \Pi_{\Hit}=\bigoplus_b |b\rangle\!\langle b|\otimes\Pi_b.
 \label{eq:app-hit-projector}
\end{equation}
It is a mathematical projector only.  It is not measured by the experiment.
By definition,
\begin{equation}
 \beta_{\mathsf{cut}}
 =\|\Pi_{\Hit}|\Psi_\Theta\rangle\|,
 \qquad
 \beta_{\mathsf{cut}}^2
 =\sum_b\|\Pi_b|\psi_b\rangle\|^2.
 \label{eq:app-beta}
\end{equation}
Thus $\beta_{\mathsf{cut}}^2$ is exactly the probability that an analytic measurement of the compressed database at the cut would find a fresh target opening.  No such measurement is performed. The direct sum is also why oracle dependence of $C$ causes no difficulty. On a fixed branch, $C$ is an ordinary immutable set.  The transition bound below is uniform over all $C$ of size at most $K$ and over all databases consistent with that branch.  Averaging therefore preserves arbitrary correlation created before the cut between $C$ and $D$.

\subsection{Transition bound for one call to the hit relation}
\label{app:one-call}

Fix a branch $b$ and a database $D$ for which $R_b(D)=0$.  For any possible query input $x$, define the local range set
\[
 L_{b,D,x}
 =\{y\in\mathcal Y:
     \Fresh_b(x)=1\ \wedge\ y\in C\}.
\]
If changing the database at the queried point can create $R_b$, the newly present pair must be $(x,y)$ with $y\in L_{b,D,x}$.  Moreover,
\[
 |L_{b,D,x}|\le |C|\le K.
\]
The new pair itself witnesses membership in $R_b$, so the predicate is strongly recognizable in the sense required by Lemma~\ref{lem:app-transition}.  Hence, with
\begin{equation}
 \alpha=\sqrt{\frac{c_{\mathsf{co}}K}{N}},
 \label{eq:app-alpha}
\end{equation}
every branch satisfies
\[
 \|\Pi_b O_{\mathsf{co}}(I-\Pi_b)\|_{\mathrm{op}}
 \le\alpha.
\]
Taking the direct sum over branches gives
\begin{equation}
 \|\Pi_{\Hit}O_{\mathsf{co}}(I-\Pi_{\Hit})\|_{\mathrm{op}}
 \le\alpha.
 \label{eq:app-global-transition}
\end{equation}

Equation~\eqref{eq:app-global-transition} remains valid for a call controlled by the branch.  On branches on which no evaluation occurs, the call is the identity and has zero transition capacity.  This observation handles a variable number of evaluations after the cut.

\subsection{Telescoping through the execution after the cut}
\label{app:telescoping}

List the oracle evaluations after the cut in execution order and pad every
branch to $Q$ positions using controlled identities. Allocate all purifying
ancillas at the start and extend each oracle-free isometry to a unitary on the
fixed register space. Let $V_0$ be the computation from the cut to the first
oracle position. For $1\le i<Q$, let $V_i$ be the computation after the $i$th
oracle call and before the $(i+1)$st call. Let $V_Q$ be the terminal
oracle-free computation after the final oracle call, including copying its
answer, completing validation, applying the stopping rule, and recording the
acceptance result. Halted branches are padded by controlled identities. None
of these maps modifies the frozen cut registers or the compressed database, so
\[
 [V_i,\Pi_{\Hit}]=0
 \qquad\text{for }0\le i\le Q.
\]

Set
\[
 |\Phi_0\rangle=|\Psi_\Theta\rangle,
 \qquad
 |\Phi_i\rangle=O_{\mathsf{co}}V_{i-1}|\Phi_{i-1}\rangle
 \quad(1\le i\le Q),
\]
and put
\[
 a_i=\|\Pi_{\Hit}|\Phi_i\rangle\|.
\]
Using the triangle inequality, unitarity of the compressed oracle call, and Equation~\eqref{eq:app-global-transition},
\begin{align*}
 a_i
 &=\bigl\|\Pi_{\Hit}O_{\mathsf{co}}V_{i-1}
       (\Pi_{\Hit}+I-\Pi_{\Hit})|\Phi_{i-1}\rangle\bigr\|\\
 &\le
   \bigl\|\Pi_{\Hit}O_{\mathsf{co}}\Pi_{\Hit}
       V_{i-1}|\Phi_{i-1}\rangle\bigr\|\\
 &\quad+
   \bigl\|\Pi_{\Hit}O_{\mathsf{co}}(I-\Pi_{\Hit})
       V_{i-1}|\Phi_{i-1}\rangle\bigr\|\\
 &\le a_{i-1}+\alpha.
\end{align*}
Here and below an isometry can be extended to a unitary without changing any norm, so the same calculation covers CPTP maps purified with a Stinespring dilation. Induction and Equation~\eqref{eq:app-beta} yield
\begin{equation}
 \|\Pi_{\Hit}|\Phi_Q\rangle\|
 \le \beta_{\mathsf{cut}}+Q\alpha.
 \label{eq:app-telescope}
\end{equation}
Define the state after all validation and output processing by
\[
 |\widetilde\Phi_Q\rangle=V_Q|\Phi_Q\rangle.
\]
Since $V_Q$ commutes with $\Pi_{\Hit}$,
\[
 \|\Pi_{\Hit}|\widetilde\Phi_Q\rangle\|
 =\|\Pi_{\Hit}|\Phi_Q\rangle\|.
\]

This step explains the lifetime query accounting.  Neither a classical ledger
call nor finalization projects away the hit component.  Finality only creates a
frozen classical branch label.  The adversary's remaining quantum workspace
and the compressed database continue unchanged.

\subsection{From a recorded hit to a winning execution}
\label{app:output}

Let $W$ be the winning event in Lemma~\ref{lem:cut}.  On a cut branch $b$, the
retained validation history up to the first winning acceptance supplies at
least one of the at most $\ell$ checked points $x_j$ for which
\[
 \Fresh_b(x_j)=1,
 \qquad H(x_j)\in C.
\]
Each check evaluation is one of the $Q$ calls after the cut and its classical answer
$y_j=H(x_j)$ is retained.  If the exact retained pair is present in the final
compressed database, then on a winning check branch that database satisfies
$R_b$.  Lemma~\ref{lem:app-recording}, followed by
Equation~\eqref{eq:app-telescope}, therefore gives
\begin{align*}
 \sqrt{\Pr[W]}
 &\le
 \|\Pi_{\Hit}|\widetilde\Phi_Q\rangle\|
 +\sqrt{\frac{2\ell}{N}}\\
 &\le
 \beta_{\mathsf{cut}}
 +Q\sqrt{\frac{c_{\mathsf{co}}K}{N}}
 +\sqrt{\frac{2\ell}{N}}.
\end{align*}
Substituting $N=2^\lambda$ proves the first display of Lemma~\ref{lem:cut}.  Finally, for nonnegative $a,b,c$, $(a+b+c)^2\le3(a^2+b^2+c^2)$.  Squaring the last inequality gives
\[
 \Pr[W]
 \le
 3\beta_{\mathsf{cut}}^2
 +\frac{3c_{\mathsf{co}}KQ^2}{2^\lambda}
 +\frac{6\ell}{2^\lambda},
\]
which is the second display.

\subsection{Why adaptive stopping does not add a loss}
\label{app:adaptive-stopping}

For completeness, we isolate where adaptivity enters the preceding proof. The stopping time can depend on every public event before the cut, including validation outcomes that depend on the oracle and adversarial scheduling.  This only changes the amplitudes and internal states in Equation~\eqref{eq:app-cut-direct-sum}.  Once a branch is fixed, its $(z,\omega,C)$ are fixed as well, and the bound for one call is uniform over the branch.  The global projector and every map after the cut are block diagonal in the retained classical history.  Their relevant operator norm is therefore the maximum of the branchwise norms, which is at most $\alpha$. Equivalently, one may condition on each classical cut branch, apply the same bound, and combine the resulting subnormalized vectors in Euclidean norm. No union bound over possible cut times is needed.  The stopping time requirement is essential, a cut selected retrospectively using a future winning output would amount to postselection and is not covered.

\subsection{Lifetime accounting}
\label{app:lifetime-proof}

We now prove Corollary~\ref{cor:lifetime}.  Index every possible protected episode by $j$.  For an adaptively absent episode, take $K_j=0$ and let its winning event be empty.  For episode $j$, let $W_j$ be the winning event that the first fresh hit attributed to that episode is accepted, let $\beta_j$ be its bad norm at the cut, and suppose
\[
 \beta_j^2\le\varepsilon_{\mathsf{pre},j}.
\]
The total number of oracle evaluations after that episode's cut is at most the conservative lifetime bound $Q_{\mathsf{life}}$.  Applying Lemma~\ref{lem:cut} with $K=K_j$, $Q=Q_{\mathsf{life}}$, and $\ell=\ell_j$ gives
\[
 \Pr[W_j]
 \le
 3\varepsilon_{\mathsf{pre},j}
 +\frac{3c_{\mathsf{co}}K_jQ_{\mathsf{life}}^2}{2^\lambda}
 +\frac{6\ell_j}{2^\lambda}.
\]
Every fresh lifetime hit lies in $\bigcup_j W_j$.  A union bound therefore gives
\begin{align*}
 \Pr[\mathsf{win}]
 &\le \sum_j\Pr[W_j]\\
 &\le
 3\sum_j\varepsilon_{\mathsf{pre},j}
 +\frac{3c_{\mathsf{co}}Q_{\mathsf{life}}^2}{2^\lambda}
      \sum_j K_j\\
 &\quad +\frac{6}{2^\lambda}\sum_j\ell_j\\
 &=
 3\varepsilon_{\mathsf{pre}}^{\mathsf{life}}
 +\frac{3c_{\mathsf{co}}K_{\mathsf{life}}
 Q_{\mathsf{life}}^2}{2^\lambda}\\
 &\quad +\frac{6\ell_{\mathsf{life}}}{2^\lambda},
\end{align*}
as claimed. If the realized number or size of episodes is random, the display
uses public deterministic upper bounds on $\sum_jK_j$ and $\sum_j\ell_j$. We do
not condition on the complete future episode schedule, since that could
postselect on later oracle-dependent events. An expectation-based refinement
would require a separate predictable variable-capacity lemma and is not claimed
here.

\subsection{Scope of the lemma}
\label{app:cut-scope}

The proof uses all of the following conditions.
\begin{enumerate}[leftmargin=1.6em]
  \item The cut is a stopping time, rather than a time selected using a
  future success event. Every winning branch reaches the cut before its
  witnessing checks, and a branch that never reaches the cut cannot win.
  \item The target register $C$ and the relation parameters $(z,\omega)$ are
  classical and fixed at the cut. They may be correlated with all oracle
  activity before the cut, and $\omega$ may remain secret until later, but
  $\omega$ is sampled and stored no later than the cut.
  \item The predicate
  \[
   \Fresh(z,\omega,\cdot)
  \]
  is deterministic, oracle-independent, and fixed after the cut. Every later
  acceptance-relevant eligible target remains in $C$.
  \item Every oracle evaluation after the cut capable of affecting the relation,
  including calls made by the verifier and retained validation or output checks,
  is included in
  $Q$.
  \item A concrete construction separately proves a bound on
  $\beta_{\mathsf{cut}}^2$.  The finalized target lemma does not make that
  contribution from before the cut negligible by itself.
\end{enumerate}
Dropping any of Items 1--4 invalidates the transition argument.  Item 5 is the reason the main theorem carries a term specific to the protocol for the contribution before the cut rather than claiming that closure alone turns every hash commitment into an object secure against preimage attacks.

\section{Full proofs for commit, close, reveal}
\label{app:ccr-proofs}

\subsection{Marked slice exposure}
\label{app:ccr-exposure}

\begin{proof}[Full proof of Lemma~\ref{lem:head-exposure}]
Fix a cell $u$ and reduce to unstructured search for its uniform secret. The
search challenger provides coherent access to
\[
 E_u(s)=\mathbf 1[s=s_u].
\]
The reduction simulates every other cell locally. For the target cell, sample
$h_u\leftarrow\bits^{\lambda_h}$, independent random oracles $G_0,G_1$, and
an independent random function $P_u$ from canonical secret deleted opening
descriptions to $\bits^{\lambda_c}$. Write
\[
 q_u(s)=\mathsf{Enc}(\mathtt{head};\ctx_u,s).
\]
Every well-formed $H_1$ input for the cell is uniquely
$x=\mathsf{Enc}_u(\varphi,s)$. Define the simulated interfaces from the outset
by
\begin{align*}
 H_0(q_u(s))
  &=\begin{cases}h_u,&E_u(s)=1,\\G_0(q_u(s)),&E_u(s)=0,
    \end{cases}\\
 H_1(\mathsf{Enc}_u(\varphi,s))
  &=\begin{cases}P_u(\varphi),&E_u(s)=1,\\
      G_1(\mathsf{Enc}_u(\varphi,s)),&E_u(s)=0.
    \end{cases}
\end{align*}
Malformed inputs and inputs for other cells use the corresponding independent
random oracle.

For fixed $s_u$, the truth table map from $(P_u,G_1)$ to $H_1$ and the unused
marked values of $G_1$ is a basis permutation. The marked slice of $H_1$ is
$P_u$, and every unmarked point is the corresponding $G_1$ value. The
analogous statement holds for $(h_u,G_0)$ and $H_0$. Thus the visible
interfaces are exactly independent random oracles. The routing is implemented
coherently by computing $E_u(s)$, selecting the appropriate oracle, and
uncomputing the selector. Every prescribed honest commitment is evaluated
through the same $P_u$ oracle as a routed adversarial evaluation, so consistency
holds even if an earlier adversarial query overlaps that input. No adaptive
programming or database deletion occurs.

Compare this routed execution with the identical circuit in which every call
to $E_u$ is replaced by the all zero oracle. Prescribed honest evaluations
remain calls to $P_u$ in both circuits. In the all zero execution, no other
candidate carrying $H_1$ call reaches $P_u$. Consequently, every coordinate of
$D_{P,u}$ was queried by a prescribed honest evaluation, whose action belongs
to $\mathcal R_u$, and
\[
 \Pi_u^P|\Psi_{\Theta_u}^{\,0}\rangle=0.
\]
Therefore
\[
 \widehat\beta_u
 \le
 \bigl\|
 |\Psi_{\Theta_u}^{\,s_u}\rangle-
 |\Psi_{\Theta_u}^{\,0}\rangle
 \bigr\|.
\]
If the routed simulation makes at most $T$ equality oracle calls, the standard
BBBV hybrid bound~\cite{bennett1997strengths} gives
\[
 \mathbb E_{s_u}
 \bigl\|
 |\Psi_{\Theta_u}^{\,s_u}\rangle-
 |\Psi_{\Theta_u}^{\,0}\rangle
 \bigr\|^2
 \le \frac{4T^2}{2^{\lambda_s}}.
\]
Computing and uncomputing the routing bit uses at most two equality calls per
charged candidate-carrying evaluation, so $T\le2R$. The close may be adaptive.
Retain its classical flag and pad every branch to the public cap before applying
the hybrid bound. Thus, for a universal constant $c_{\mathsf{eq}}$,
\[
 \sum_u\widehat\beta_u^2
 \le c_{\mathsf{eq}}
 \frac{N(R+1)^2}{2^{\lambda_s}}.
\]
This proves~\eqref{eq:head-exposure-eq}.

For the post-cut exact-secret argument, define
\[
 \Fresh^P_{u,z}(\varphi)
 =\Fresh(z,(s_u,\mathcal R_u),\mathsf{Enc}_u(\varphi,s_u))
\]
and let $\Hit^P_{u,z,C}(D_{P,u})=1$ exactly when some
$\varphi\in\operatorname{dom}(D_{P,u})$ satisfies
$\Fresh^P_{u,z}(\varphi)=1$ and $D_{P,u}(\varphi)\in C$. Equality routing, the
off-slice $G_1$ call, and every oracle-free map act trivially on $D_{P,u}$.
Only the controlled $P_u$ call can create this relation. Since at most
$|C|\le K$ range values can create it, the compressed-oracle transition bound
for each charged logical $H_1$ evaluation is
\[
 \bigl\|\Pi_{\Hit^P}O_{H_1}^{\mathsf{route}}
 (I-\Pi_{\Hit^P})\bigr\|_{\mathsf{op}}
 \le \sqrt{\frac{c_{\mathsf{co}}K}{2^{\lambda_c}}}.
\]
A prescribed honest $P_u$ evaluation has zero transition capacity because its
action is not fresh. On an exact secret final check, the verifier retains
$(\varphi,P_u(\varphi))$, so the output-recording inequality applies with
residual $\sqrt{2\ell/2^{\lambda_c}}$. The finalized target projector also
requires $P_u(\varphi)\in C_u$ and is therefore dominated by $\Pi_u^P$.

For alternatives, conditioned on $(s_u,h_u)$, every
$H_0(q_u(s))$ with $s\ne s_u$ is independent and uniform in
$\bits^{\lambda_h}$. The event $\mathsf{AltHead}_R$ requires a retained or
output candidate to be checked and found to match. Quadratic search for a
random target, including the final check, gives the second bound after a union
bound over cells. Revealing $s_u$ identifies only the excluded planted point
and does not change the distribution outside it.
\end{proof}

\subsection{Completeness under bounded contention}

Recall that \(\delta_{\mathsf{slot}}\in\{0,1\}\) is the rounding needed to move
from \(d_i\) to the first slot strictly after \(d_i\).  In continuous time
notation, or when the protocol treats the boundary as the start of the next
slot, \(\delta_{\mathsf{slot}}=0\).

\begin{proof}[Full proof of Proposition~\ref{prop:ccr-complete}]
The request arrives while \(u=(a,e,i)\) is the live cell and the signer emits its commitment in slot \(t_{\mathsf{req}}\).  Under \((\Delta_{\mathsf{inc}},\Delta_{\mathsf{fin}})\) liveness, it is included by \(t_{\mathsf{req}}+\Delta_{\mathsf{inc}}\) and belongs to a finalized prefix by
\[
 t_C\le t_{\mathsf{req}}+\Delta_{\mathsf{inc}}
                       +\Delta_{\mathsf{fin}}<d_i.
\]
The proposition assumes
\(t_i^{\mathsf{open}}<t_{\mathsf{req}}\), and an event cannot finalize before
it is emitted, so
\(t_i^{\mathsf{open}}<t_{\mathsf{req}}\le t_C\).  By hypothesis, fewer than
$M$ distinct candidates precede $c_i$ in the canonical admission order.  Hence
the cap does not exclude it and the honest digest \(c_i\) belongs to \(C_u\)
before the close.  Once this fact is finalized, monotone finality preserves it.
The signer does not expose \(s_i\) until both this admission fact is final and the
close has passed.  At the first admissible slot
\(t_{\mathsf{rev}}\le d_i+\delta_{\mathsf{slot}}\), it emits the reveal.  That
event satisfies all four transition checks.  The cell and canonical guards are
current by the guard-stability premise. Its inclusion is after \(d_i\).
Construction gives
\(H_0(\mathsf{Enc}(\mathtt{head};\ctx_u,s_i))=h_i\). Recomputing under \(H_1\) returns
\(c_i\in C_u\).  Liveness therefore includes and finalizes the reveal by
\[
 d_i+\delta_{\mathsf{slot}}
       +\Delta_{\mathsf{inc}}+\Delta_{\mathsf{fin}}.
\]
Serialization ensures that no second honestly requested action can consume the
cell. If the requested receipt is not the first receipt for the cell, a
different action has produced a fresh accepting receipt first. This proves the
stated dichotomy. The probability bound follows by adding liveness, finality,
state-execution, and $\LAEUF$ failure. Filling all $M$ positions can deny
service without violating authentication safety, which is why the admission
rank premise is necessary.
\end{proof}

\subsection{Security games and lifetime accounting}

The exposure argument before the cut charges the joint sequence of oracle
queries that can carry candidates.  An \(H_0\) query can concentrate amplitude on a secret consistent
with a public head, after which one \(H_1\) query can carry that candidate.
Conversely, an \(H_1\) query can carry a blind candidate tested only later.
Recall
\begin{equation}
 \Qhead:=Q_0+Q_1.
 \label{eq:q-head}
\end{equation}
The counts are deterministic lifetime caps and include direct, mediated,
relevant honest, and evaluations used to check output.  Their sum is a conservative
budget for candidate carrying queries in Lemma~\ref{lem:head-exposure}.  We also recall
\begin{equation}
 \Delta_{\mathsf{PRF}}
 :=\sum_{a=1}^{U_{\mathsf{life}}}
       \Adv_F^{\qPRFQRO}(\mathcal B_a).
 \label{eq:multi-prf-loss}
\end{equation}
Under the single-key $\qPRFQRO$ convention,
\[\Delta_{\mathsf{PRF}}\le
U_{\mathsf{life}}\max_a\Adv_F^{\qPRFQRO}(\mathcal B_a).\]
Each reduction stops if its target account is corrupted, so it never has to
reveal a challenge key after replacing its outputs.

In the remainder, \(N_a,K_a,\ell_a\) are the public deterministic
quotas for each account fixed in the theorem, and the lifetime quantities are
their sums.  This allocation is what permits numerical bounds proved in
different hybrids for target accounts to be added.

\begin{proof}[Full proof of Theorem~\ref{thm:ccr}]
We use the following games.

\paragraph{Game \(G_0\).}
This is the real \(\LAEUF\) experiment for \(\CR\) in the \(\LQROM\).

\paragraph{Game \(G_1\).}
Abort on
\(\Bad_{\mathsf{fin}}\vee\Bad_{\mathsf{can}}
\vee\Bad_{\mathsf{state}}\).  Equation~\eqref{eq:ledger-bad-union} gives
\begin{equation}
 \Pr[G_0\text{ wins}]
 \le \Pr[G_1\text{ wins}]
 +\varepsilon_{\mathsf{fin}}+
  \varepsilon_{\mathsf{can}}+
  \varepsilon_{\mathsf{state}}.
 \label{eq:ccr-g01}
\end{equation}

\paragraph{Game \(G_2\).}
Let \(E_{\mathsf{acc}}\) be the event that the transition machine accepts an
action outside the current requested set for an honest account uncorrupted at
that step. Every win in $G_1$ implies $E_{\mathsf{acc}}$. Let \(W_a\) be the
event that the first such transition is under account \(a\). For each \(a\),
define \(G_2^{(a)}\) by replacing
\(F_{k_a}(\mathsf{Enc}(\mathtt{secret};\ctx_u))\) for cells of \(a\), up to corruption of
\(a\), by independent uniform \(\lambda_s\) bit strings.  Reduction
\(\mathcal B_a\) answers derivation calls for \(a\) with its $\qPRFQRO$ challenge
oracle, simulates other accounts with locally generated real keys, and stops at
the first of corruption of \(a\) or resolution of \(W_a\). On \(W_a\), its
defining transition necessarily occurs before corruption. Define
\[
 p_{2,\mathsf{acc}}
 :=\sum_a\Pr_{G_2^{(a)}}[W_a].
\]
The events \(W_a\) partition $E_{\mathsf{acc}}$, so summing the distinguishing
gaps gives
\begin{equation}
 \Pr[G_1\text{ wins}]
 \le \Pr_{G_1}[E_{\mathsf{acc}}]
 \le p_{2,\mathsf{acc}}+\Delta_{\mathsf{PRF}}.
 \label{eq:ccr-g12}
\end{equation}
This handles arbitrary corruptions of other users and arbitrary corruption
times.  It also avoids the invalid step of returning a PRF challenge key after
having replaced its outputs by a random function. All estimates below are
performed in the unconditioned experiment $G_2^{(a)}$. The event $W_a$ only
restricts the event whose probability is bounded; the state and random oracles
are never conditioned on $W_a$. We bound
$\Pr_{G_2^{(a)}}[W_a]$ using only protected cells and checks of account $a$,
and then sum over $a$. The deterministic quotas
\(N_a,K_a,\ell_a\) sum to the displayed lifetime totals even though the hybrids
are distinct probability spaces. Fix the selected accepting transition and write
\(u=(a,e,i)\) for its cell.  Its complete canonical action \(\tau^\star\)
includes the successor head and is outside \(\Requested[a]\) at its acceptance
step. Because state execution is sound, \(u\) was live immediately
before the reveal and the accepted reveal supplies \((s^\star,r^\star)\) such
that
\begin{align}
 H_0\bigl(\mathsf{Enc}(\mathtt{head};\ctx_u,s^\star)\bigr)&=h_i,
 \label{eq:ccr-forge-head}\\
 H_1(x^\star)&\in C_u,
 \label{eq:ccr-forge-commit}
\end{align}
where
\[
 x^\star=
 \mathsf{Enc}(\mathtt{commit};
 \ctx_u,d_i,\tau^\star,s^\star,r^\star).
\]

First split the selected acceptance according to whether \(s^\star=s_u\). If
\(s^\star\ne s_u\), the accepted reveal is an alternative preimage of the
planted public head.  Apply Lemma~\ref{lem:head-exposure} to the at most \(N_a\)
cells of account \(a\) in each target hybrid, then sum over \(a\).  Since
\(\sum_aN_a=N_{\mathsf{life}}\), this gives
\begin{equation}
 p_{2,\mathsf{alt}}
 :=\sum_a\Pr_{G_2^{(a)}}[W_a\wedge s^\star\ne s_u]
 \le c_{\mathsf{alt}}
 \frac{N_{\mathsf{life}}(\Qhead+1)^2}
      {2^{\lambda_h}}.
 \label{eq:ccr-alt-bound}
\end{equation}

It remains to analyze selected acceptances with \(s^\star=s_u\). Within target hybrid \(a\),
for every protected cell \(u\) of account \(a\), take the stopping time
\(\Theta_u=d_i\). At that time the finalized
transcript deterministically copies \(C_u\) to a read-only classical register.
Eligibility is defined to be membership in this set after \(d_i\), so every
later acceptance-relevant target remains in \(C_u\). This is a closed finalized cut in the
sense of Definition~\ref{def:cut}.  The random \(s_u\) was sampled and stored
in a classical honest state register before this cut, although it is not
disclosed until later. The set $\mathcal R_u$ is also frozen in classical game
state at the cut. Use $\omega=(s_u,\mathcal R_u)$ as the hidden parameter of
Lemma~\ref{lem:cut}. The pre-close authorization gate ensures that every
prescribed honest commitment action for this cell belongs to $\mathcal R_u$.

Define \(\Fresh_u\) at the cut to require a well formed opening for \(u\), secret
coordinate exactly \(s_u\), and a complete action outside the frozen value of
\(\Requested[a]\). These conditions remain fixed thereafter. A winning action
is outside \(\Requested[a]\) at its acceptance step, and hence also at the
earlier cut.
Omitting later guards on the current state only enlarges the analytic fresh relation.
Let \(\beta_u\) be the bad norm at the cut for the marked-slice $P_u$
database and this fresh relation for the exact secret. If its hit projector is
nonzero, that database already contains a designated fresh description on the
hidden \(s_u\)-slice. No
honest process reveals \(s_u\) before \(d_i\).
Equation~\eqref{eq:head-exposure-eq}, first with \(N_a\) in each target hybrid
and then summed over \(a\), therefore gives
\begin{equation}
 \sum_u\beta_u^2
 \le \sum_u\widehat\beta_u^2
 \le c_{\mathsf{eq}}
 \frac{N_{\mathsf{life}}(\Qhead+1)^2}{2^{\lambda_s}}.
 \label{eq:ccr-pre-bound}
\end{equation}
This covers both a blind \(H_1\) carrier query and a search assisted by \(H_0\)
followed by an \(H_1\) carrier query.  Using \(Q_1\) alone here would be
invalid. No measurement of the compressed database is made. On the
complementary amplitude, the adversary retains its entire quantum state from
before the cut and may subsequently learn \(s_i\) from a censored honest
reveal. Nevertheless, closure prevents any new digest from becoming eligible
outside \(C_u\).
Equations~\eqref{eq:ccr-forge-head}--\eqref{eq:ccr-forge-commit} show that a
selected exact-secret acceptance must supply a checked fresh input whose \(H_1\)
image lies in one of the immutable lifetime target sets.  Apply
Corollary~\ref{cor:lifetime} inside each target hybrid with \(K_a,\ell_a\) and
the lifetime budget \(Q_1\), then sum over \(a\).  We obtain
\begin{align}
 p_{2,\mathsf{acc}}
 \le{}&p_{2,\mathsf{alt}}\\
 &+3\sum_u\beta_u^2
 +\frac{3c_{\mathsf{co}}K_{\mathsf{life}}Q_1^2}
        {2^{\lambda_c}}
 +\frac{6\ell_{\mathsf{life}}}{2^{\lambda_c}}.
 \label{eq:ccr-cut-bound}
\end{align}
The use of a lifetime budget permits arbitrary interleaving of cells and
retention of quantum state across every close.  Every mediated opening check up
to the first fresh acceptance is included in \(Q_1\), and its retained pair is
included in \(\ell_{\mathsf{life}}\).

Combining~\eqref{eq:ccr-g01}, \eqref{eq:ccr-g12},
\eqref{eq:ccr-alt-bound}, \eqref{eq:ccr-pre-bound}, and
\eqref{eq:ccr-cut-bound}, and upper bounding both head search denominators by
\(2^{\min\{\lambda_s,\lambda_h\}}\), yields
\begin{align*}
 \Adv_{\CR}^{\LAEUF}(\mathcal A)
 \le{}&\Delta_{\mathsf{PRF}}
 +\varepsilon_{\mathsf{fin}}
 +\varepsilon_{\mathsf{can}}
 +\varepsilon_{\mathsf{state}}\\
 &+(c_{\mathsf{alt}}+3c_{\mathsf{eq}})
 \frac{N_{\mathsf{life}}(\Qhead+1)^2}
      {2^{\min\{\lambda_s,\lambda_h\}}}\\
 &+3c_{\mathsf{co}}
 \frac{K_{\mathsf{life}}Q_1^2}{2^{\lambda_c}}
 +\frac{6\ell_{\mathsf{life}}}{2^{\lambda_c}}.
\end{align*}
Thus one may take \(c_h=c_{\mathsf{alt}}+3c_{\mathsf{eq}}\) and
\(c_c=3c_{\mathsf{co}}\).  The final term for checked output has the explicit
coefficient six. Under the multi-key $\qPRFQRO$ convention
following~\eqref{eq:multi-prf-loss} and the convention for candidate carrying queries
in Equation~\eqref{eq:q-head}, this is the displayed bound of
Theorem~\ref{thm:ccr}.
\end{proof}

\end{document}